%% file: EndToEndScheduling_TCPS19.tex
\documentclass[format=acmsmall, review=false, screen=true]{acmart}

\usepackage{booktabs} 
\usepackage{epstopdf}

\usepackage{pgfplots}
\pgfplotsset{compat=newest}
\pgfplotsset{plot coordinates/math parser=false}

\usepackage{caption}

\usepackage[tight]{subfigure}

\graphicspath{{./Figures/}}

\usepackage{color}
\newtheorem{problem}{Problem}

\newtheorem{remark}{Remark}

\usepackage{mathtools}
\DeclarePairedDelimiter\ceil{\lceil}{\rceil}
\DeclarePairedDelimiter\floor{\lfloor}{\rfloor}

\newcommand{\figref}[1]{Fig.~\ref{fig:#1}}

\newcommand{\thmref}[1]{Theorem~\ref{thm:#1}}

\usepackage{colortbl}

\newcommand{\LE}[1]{$\blacktriangledown$\footnote{LESI: #1}}

\newcommand{\PA}[1]{$\spadesuit$\footnote{MP: #1}}

\setcopyright{none}

\makeatletter
\def\@mkbibcitation{\relax}
\makeatother

\begin{document}

\title{Integrating Security in Resource-Constrained Cyber-Physical Systems}
\titlenote{This work was supported in part by the NSF CNS-1652544 and CNS-1505701 grants, and the Intel-NSF Partnership for Cyber-Physical Systems Security and Privacy. This material is also based on research sponsored by the ONR under agree- ments number N00014-17-1-2012 and N00014-17-1-2504.}

\author{Vuk Lesi}
\affiliation{
  \institution{Duke University}
  \city{Durham}
  \state{North Carolina}
  \country{USA}
}
\email{vuk.lesi@duke.edu}

\author{Ilija Jovanov}
\affiliation{
  \institution{Duke University}
  \city{Durham}
  \state{North Carolina}
  \country{USA}
}
\email{ilija.jovanov@duke.edu}

\author{Miroslav Pajic}
\affiliation{
  \institution{Duke University}
  \city{Durham}
  \state{North Carolina}
  \country{USA}
}
\email{miroslav.pajic@duke.edu}

\renewcommand{\shortauthors}{}

\begin{abstract}
Defense mechanisms against network-level attacks are commonly based on the use of cryptographic techniques, such as lengthy message authentication codes (MAC) that provide data integrity guarantees. However, such mechanisms require significant resources (both computational and network bandwidth), which  prevents their continuous use in resource-constrained cyber-physical systems (CPS).
Recently, it was shown how physical properties of controlled systems can be exploited to relax these stringent requirements for systems where sensor measurements and actuator commands are transmitted over a potentially compromised network; specifically, that merely intermittent use of data authentication (i.e., at occasional time points during system execution), can still  provide strong Quality-of-Control (QoC) guarantees even in the presence of false-data injection attacks, such as \emph{Man-in-the-Middle} (MitM) attacks. Consequently, in this work we focus on integrating security into existing resource-constrained CPS, in order to protect against MitM attacks on a system where a set of control tasks communicates over a real-time network with system sensors and actuators. We introduce a design-time methodology that incorporates requirements for QoC in the presence of attacks into end-to-end timing constraints for real-time control transactions, which include data acquisition and authentication, real-time network messages, and control tasks. This allows us to formulate a mixed integer linear programming-based method for direct synthesis of schedulable task and message parameters (i.e., deadlines and offsets) that do not violate timing requirements for the already deployed controllers, while adding a sufficient level of protection against network-based attacks; specifically, the synthesis method also provides suitable intermittent authentication policies that ensure the desired QoC levels under attack. 
To additionally reduce the security-related  bandwidth overhead, we propose the use of cumulative message authentication at time instances when integrity of messages from subsets of sensors should be ensured. Furthermore, we introduce a method for opportunistic use of remaining resources to further improve the overall QoC guarantees while ensuring system (i.e., task and message) schedulability. Finally, we demonstrate applicability and scalability of our methodology on synthetic automotive systems as well as a real-world automotive case-study.

\end{abstract}



\maketitle

\input{intro}
\input{motivation}

\input{modeling}

\input{scheduling}
\input{MILP}

\input{opportunistic}

\input{evaluation}

\input{relatedWork}

\input{conclusion}

\bibliographystyle{ACM-Reference-Format}
\bibliography{Bibliography_EMSOFT17-RTSS17,Bibliography_EMSOFT18,bibliography_SecureControl,PapersCPSL}

\end{document}

%% file: intro.tex

\section{Introduction}
\label{sec:intro}


In this work, we focus on securing resource-constrained cyber-physical systems (CPS) from network-based false-data injection attacks over low-level networks used for real-time communication of control-related messages.  With these \emph{Man-in-the-Middle} (MitM) attacks, attackers can inject  maliciously crafted data into communication between sensors and controllers, forcing a controlled physical plant into a potentially unsafe state; this is achieved either directly (by injecting false control commands) or through actions of the controller (if sensor measurements are falsified). Several such attacks have been reported recently (e.g.,~\cite{car_security2010, car_security2011, stuxnet2011,stuxnet2}); 
for example, 
susceptibility of modern automotive systems to this type of attacks was  illustrated in e.g.,~\cite{wired_jeep15, car_security2011}.
These attacks are especially threatening as they enable a \emph{remote} attacker to compromise safety-critical control features of a system, by taking over some of the components with access to a low-level safety-critical network used for control, before using them to transmit malicious control-related messages. 


Protection against this type of attacks is commonly based on data integrity enforcements using message authentication. Standard  methods for ensuring authenticity of sensor data require signing of message authentication codes (MACs) on the sensor electronic control units (ECUs), transmitting sensor measurements along with the MACs, and verification of the MACs at the controller ECUs. However, due to security-related overhead this approach may not be applicable to resource-constrained embedded platforms, which are especially dominant in legacy systems. For example, our experiments on a $96~MHz$ ARM Cortex-M3-based ECU show that executing a single-input-single-output PID controller update takes approximately $5\mu s$, while signing a $128$~bit MAC over a single measurement requires around~$100~\mu s$. Thus, resource constraints may make~it~infeasible to provide continuous protection of sensing data by authenticating every transmitted~sensor measurement. Consequently, in this work we seek to answer the question exactly how much security enforcement is sufficient and how can we exploit available system resources in order to improve the overall security guarantees, in terms of Quality-of-Control (QoC) in the presence~of~attack.

Due to the recently reported security incidents, the problem of securing CPS has drawn significant attention, with research efforts focused on the impact of false-data injection attacks on system performance (mainly QoC), as well as the design of attack-detectors and attack-resilient controllers using a physical model of the system (e.g.,~\cite{pasqualetti2013attack,ncs_attack_models,fawzi_tac14,pajic_tcns17,pajic_iccps14,miao_tcns17,shoukry2018smt}). One of the main results is that even when physics-based intrusion detectors are used, by changing messages received at the controller from a subset of system sensors, an attacker could launch stealthy (i.e., non-detectable) attacks that force the plant into any undesired state through the actions of the controller~\cite{mo2010false,kwon2014stealthy,smith_decoupled_attack11}.
%

On the other hand, we have recently shown how physical properties of the controlled system under consideration,  can be exploited to relax integrity requirements for secure control~\cite{jovanov_cdc17,jovanov_arxiv17,jovanov_cdc18}. Furthemore, by computing reachable regions of the state estimation error under stealthy attacks, control performance under attack can be evaluated for intermittent integrity enforcement policies -- i.e., policies that only intermittently employ message authentication. In~\cite{lesi_tecs17}, we condense these reachable regions into \emph{QoC degradation curves} that quantify the interplay between computational (and bandwidth) requirements imposed by security services and the QoC-guarantees under attack.
However, the use of such policies introduces new challenges for ensuring timeliness of deployed control functionalities, as the standard periodical task and message models under such relaxed integrity enforcement policies feature significant execution and transmission time variations.
In~\cite{lesi_tecs17}, we only focus on the computational aspect of the problem and show how to guarantee timeliness for security-aware control tasks, while~\cite{lesi_rtss17} presents our initial attempt to ensure timeliness of communication messages. Yet, both works consider decoupled scenarios where either ECU processing time is the only concern with the assumption that the network is not congested, or where network bandwidth is the only limitation for incorporating security while ECUs are not considered. 
However, the problem of providing integrated QoC and security guarantees while ensuring timeliness in scenarios where both ECU processing time and network bandwidth are limited remains open. Moreover,
~\cite{jovanov_arxiv17} shows that block-authentication of sensor measurements has to be used for general types of dynamics of controlled physical processes,  which results in workloads that cannot be modeled within the existing framework from~\cite{lesi_tecs17}.


Consequently, in this work we introduce a design-time methodology, illustrated in \figref{methodology}, that ensures that existing control functionalities will not be negatively affected by adding message authentication to enforce data integrity. Specifically, the methodology provides sensing-to-actuation timeliness guarantees for security-aware control that employs intermittent message authentication in order to guarantee that a desired QoC level is maintained even under attack.
To capture the cases where block-authentication is needed while further reducing bandwidth requirements for the QoC guarantees under attack, we propose the use of intermittent cumulative authentication policies.
We specifically address modeling as well as capture schedulability conditions for security-aware sensing tasks (which are preemptive) that perform cumulative 
authentication, and security-aware messages (which are non-preemptive) that support arbitrary offsets; we show in Section~\ref{sec:schedulabilityAnalysis} that existing conditions do not support general offsets, which limits the use of our preliminary approach from~\cite{lesi_rtss17}. To further utilize resources available at runtime, we show how by opportunistically authenticating additional sensor measurements when computation time/bandwidth is available, we can further enhance QoC guarantees under attack. Finally, we show applicability of our approach on both synthetic systems that are designed according to established guidelines for automotive benchmarks, as well as an automotive case study.

\begin{figure}[!t]
  \centering
  \includegraphics[width=0.64\linewidth]{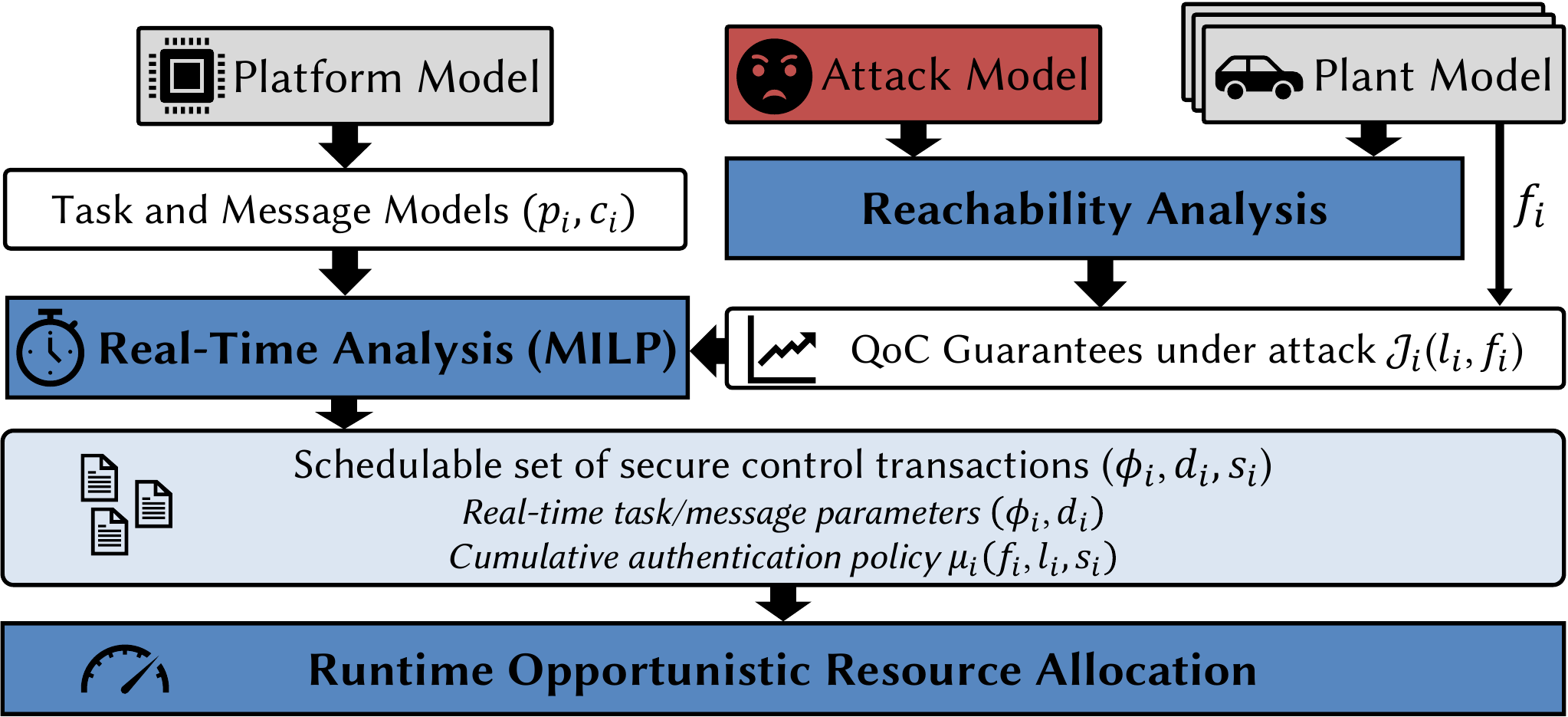}
    \vspace{-4pt}
  \caption{Design-time methodology to integrate security in resource-constrained CPS; the use of cumulative intermittent message authentication policies enables tradeoffs between (i)~required system resources (i.e., to ensure that all control functionalities still perform within specifications even after security is `added' to the system),  and (ii) Quality-of-Control (QoC) guarantees in the presence of network-based false-data injection attacks on sensor measurements delivered to controllers.}
    \vspace{-2pt}
\label{fig:methodology}
\end{figure}

This paper is organized as follows. In Section~\ref{sec:motivation}, we 
present the system and attack models, before introducing intermittent authentication policies for secure control of CPS (Section~\ref{sec:intermittent}), and formalizing the end-to-end transaction modeling for secure control (Section~\ref{sec:modeling}). Schedulability analysis pertaining to the models is presented in Section~\ref{sec:schedulabilityAnalysis}, while Section~\ref{sec:MILP} transforms the corresponding parameter synthesis problem into a mixed integer linear program (MILP). Opportunistic use of remaining resources to improve the overall QoC guarantees in the presence of attacks is presented in Section~\ref{sec:opportunistic}, before evaluating our approach in Section~\ref{sec:evaluation}. Finally, Section~\ref{sec:relatedWork} presents related work before concluding remarks are provided in Section~\ref{sec:conclusion}.

%% file: motivation.tex
\section{System and Attack Model}
\label{sec:motivation}

In this section, we present system architecture and model, including the attack model, and introduce cumulative authentication policies that ensure the desired QoC levels in the presence of attacks. We then formalize the problem of adding security guarantees against MitM attacks and outline our design-time methodology (shown in \figref{methodology}) to integrate security in resource-constrained CPS.

\begin{figure}[!t]
  \centering
  \includegraphics[width=.84\linewidth]{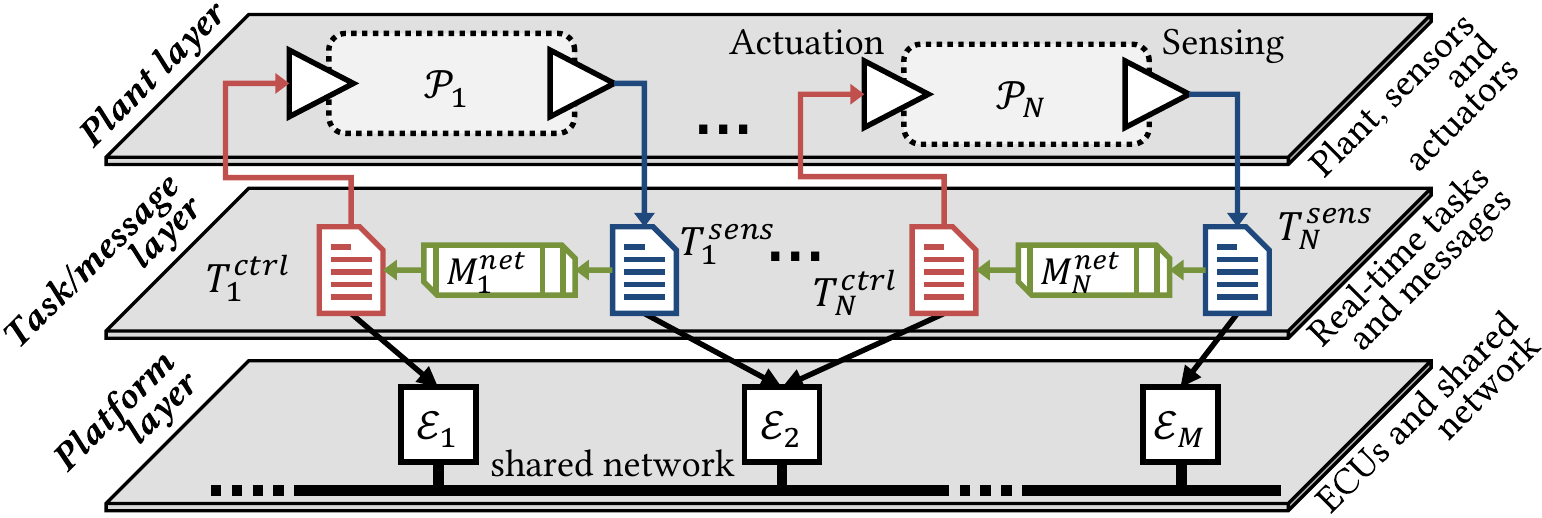}
  \caption{System architecture with $N$ physical plants ($\mathcal{P}_1,...,\mathcal{P}_N$) that are sampled and controlled in real-time by M ECUs ($\varepsilon_1,...,\varepsilon_M$); the ECUs communicate with the corresponding plants' sensors and actuators over a real-time communication network. We assume that the mapping of controllers for each plant $\mathcal{P}_i$ to a specific ECU $\varepsilon_j$ is already performed.}
\label{fig:stdArch}
\end{figure}

\subsection{System Architecture and Model without Attacks}
%

We consider a common CPS architecture shown in \figref{stdArch}, where sensors for $N$ physical plants $\mathcal{P}_i$ ($i=1,..,N$), as illustrated in the \emph{plant layer} in the figure, communicate with plant controllers over a shared real-time network.
We assume that each plant $\mathcal{P}_i$ can be modeled in the standard linear systems form as
\begin{equation}
\begin{split}
	\mathbf{x}_i[{k+1}] &= \mathbf{A}_i\mathbf{x}_i[k] + \mathbf{B}_i\mathbf{u}_i[k] + \mathbf{w}_i[k] \\
	\mathbf{y}_i[k] &= \mathbf{C}_i\mathbf{x}_i[k] + \mathbf{v}_i[k],
\end{split}
\label{eq:system}
\end{equation}
where $\mathbf{x}_i[k], \mathbf{y}_i[k]$ and $\mathbf{u}_i[k]$ denote the plant's state, output and input vectors at time $k$, while  $ \mathbf{w}_i[k]$ and $\mathbf{v}_i[k]$ are process and measurement noise.
In addition, each plant $\mathcal{P}_i$ is controlled by a feedback controller that in the most general form can be captured as
\begin{equation*}
\begin{split}
\hat{\mathbf{x}}_i[k+1] &= \mathbf{f}_i\left(\hat{\mathbf{x}}_i[k],\hat{\mathbf{y}}_i[k]\right) \\
\mathbf{u}_i[k] &= \mathbf{g}_i\left(\hat{\mathbf{x}}_i[k],\hat{\mathbf{y}}_i[k]\right).
\end{split}
\end{equation*}	
Here, $\mathbf{f}_i(\cdot)$ and $\mathbf{g}_i(\cdot)$ denote arbitrary linear mappings, which may for example describe an observer-based state feedback controller illustrated in \figref{controllerArch}. In addition, $\hat{\mathbf{x}}_i[k]$  and $\hat{\mathbf{y}}_i[k]$ denote the estimate of the plant's state and sensor measurements received by the controller at time $k$. Also, as shown in \figref{controllerArch}, we assume that each controller is equipped with a physics-based intrusion/anomaly detector that employs the plant model and a window of previous control inputs ($\mathbf{u}_i[k]$), state estimates ($\hat{\mathbf{x}}_i[k]$)  and received sensor measurements ($\hat{\mathbf{y}}_i[k]$) to trigger alarms (e.g.,~as in~\cite{pajic_tcns17,mo2010false, kwon2014stealthy,jovanov_cdc17,miao_tcns17}).

\subsubsection{Task and Message Models}
For each plant $\mathcal{P}_i$, measurement acquisition, packing and transmission is done by a periodic \emph{sensing} (or \emph{transmitting}) task denoted by
$T^{sens}_i$. In addition, periodic \emph{control} (or \emph{receiving}) task $T^{ctrl}_i$, which may be executed on a different ECU, unpacks received measurements before using them for control updates in each sampling (i.e., actuation) period. 
Hence, the periods of these tasks are equal to the sampling period of the controlled plant -- i.e.,~$p_i^{sens}=p_i^{ctrl}=p_i$. We also assume that mapping of tasks onto ECUs has already occured, as shown in \figref{stdArch} -- i.e., the set $\mathcal{T}_{\mathcal{E}_j}, j=1,...,M$, of
tasks executing on each of the $M$ ECUs $\mathcal{E}_1, ...\mathcal{E}_M$ is known; for example, in the \emph{platform layer} in~\figref{stdArch}, the task set $\mathcal{T}_{\mathcal{E}_2}$ that contains $T^{sens}_1$ and $T^{ctrl}_N$ is mapped onto ECU $\mathcal{E}_2$.
Thus, we assume that the worst-case execution times (WCET) for all these tasks are known, and let $c_i^{ctrl}$ and $c_i^{sens}$ denote the WCET
on the assigned ECUs, for tasks $T_i^{ctrl}$ and $T^{sens}_i$, ($i=1,...,N$).

Each sensing task $T^{sens}_i$ communicates sensor measurements to control task $T^{ctrl}_i$ through a real-time message $M_i^{net}$ with the same period $p_i$ and the worst-case transmission time $c_i^{net}$, as illustrated in the task/message layer in \figref{stdArch}. Note that when no confusion arises, we refer to all $T_i^{sens}$, $M_i^{net}$, and $T_i^{ctrl}$ as tasks. Finally, without loss of generality, we assume that actuation is done directly by control tasks, i.e., actuation commands are not transmitted as messages over the network, although the presented model can be easily generalized to cover this case.

\begin{figure}[!t]
\centering
\begin{minipage}[t]{.44\textwidth}
  \centering
  \includegraphics[width=\textwidth]{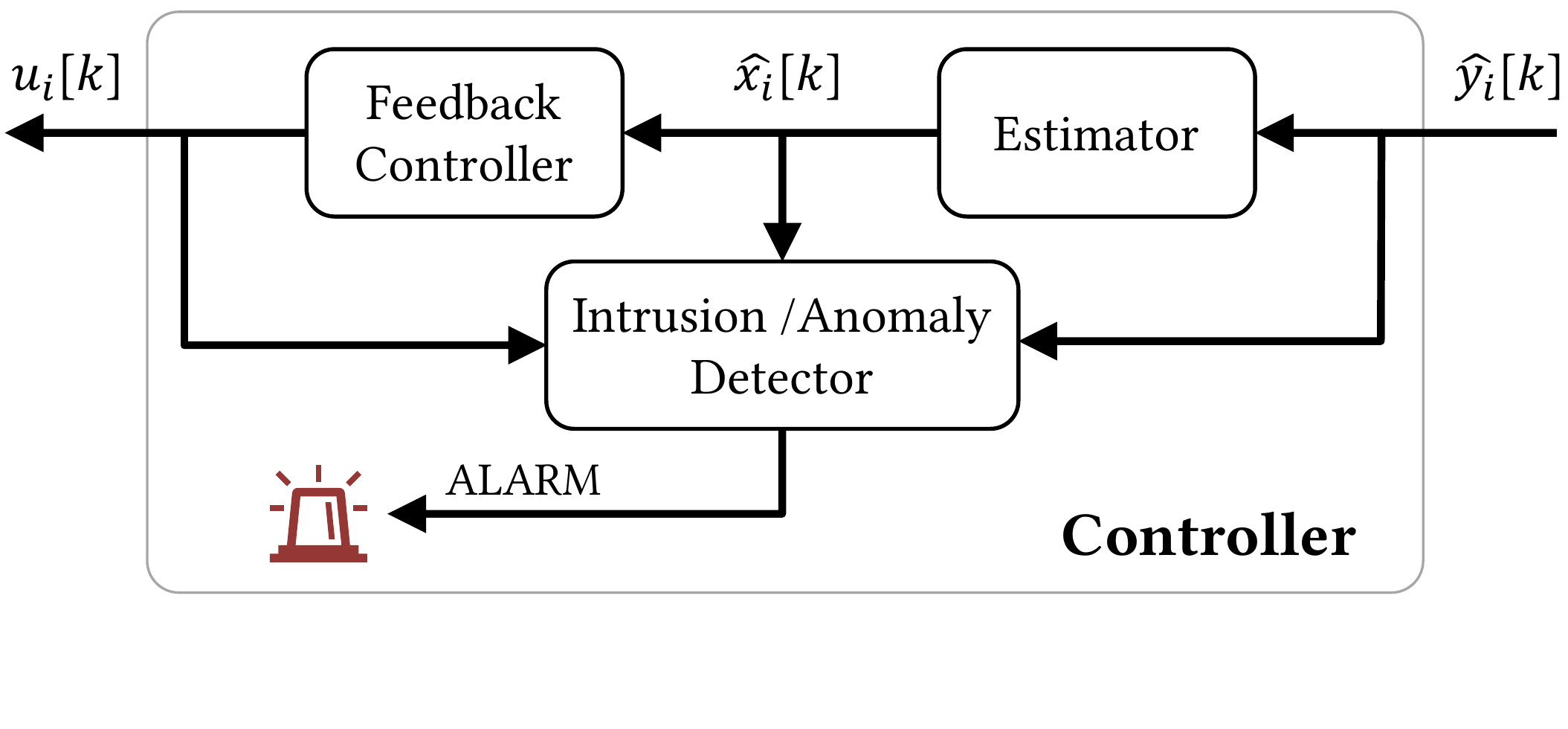}
  \captionof{figure}{General controller design. In addition to a standard estimator (i.e., observer) and a feedback controller, the controller employs a physics-based intrusion/anomaly detector.}
\label{fig:controllerArch}
\end{minipage}
\hspace{8pt}
\begin{minipage}[t]{.5\textwidth}
  \centering
  \includegraphics[width=\linewidth]{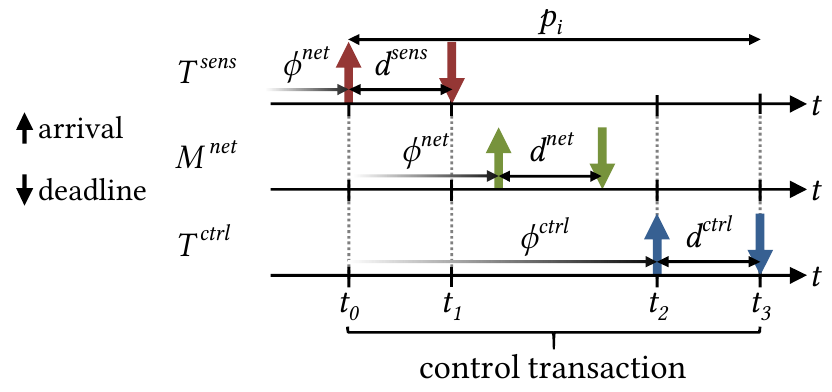}
  \captionof{figure}{Timing diagram of a control transaction --- the precedence requirements for sensing (transmitting) task $T_i^{sens}$, message $M_i^{net}$ and control (receiving) task $T^{ctrl}_i$ are captured by constraints~\eqref{eq:c1}-\eqref{eq:c3}.}
\label{fig:endToEndTiming}
\end{minipage}%
\end{figure}

\paragraph{Control Transactions}\label{sec:CT}
For any plant~$\mathcal{P}_i$, we define a \emph{control transaction} $\mathcal{T}_i$ 
as the chain of 
invocations of
$T_i^{sens}$, $M_i^{net}$ and $T_i^{ctrl}$ with all
%
%
the tasks 
being precedence-constrained. 
Specifically, the earliest time a job of 
task $T_i^{ctrl}$ may start execution is upon receiving the required sensor message. Similarly, network access for message $M_i^{net}$ cannot be requested before task $T_i^{sens}$ has prepared data for transmission. We capture these precedence constraints with non-zero offsets and constrained deadlines imposed on the tasks (\figref{endToEndTiming}); we model the tasks in the standard $(WCET, period, offset, deadline)$ format as $T_i^{ctrl}(c_i^{ctrl}, p_i,\phi_i^{ctrl},d_i^{ctrl})$, $M_i^{net}(c_i^{net}, p_i,\phi_i^{net},d_i^{net})$, and $T_i^{sens}(c_i^{sens}, p_i,\phi_i^{sens},d_i^{sens})$, with the precedence constraints specified as
\begin{align}
\label{eq:c1}
\phi_i^{net}&\geq \phi_i^{sens} + d_i^{sens},\\
\label{eq:c2}
\phi_i^{ctrl}&\geq \phi_i^{net}+d_i^{net},\\
\label{eq:c3}
\phi_i^{ctrl}+d_i^{ctrl}&\leq p_i,
\end{align}
and illustrated in \figref{endToEndTiming}.
To simplify our notation, constraint~\eqref{eq:c3} employs a standard assumption (e.g.,~as in~\cite{relaxingPeriodicityCAN}) that the delay  between sampling and actuation for each plant $\mathcal{P}_i$ is bounded by the control period~$p_i$; however, these constraints can be easily adjusted for any fixed sampling-to-actuation delay bounds that may be considered. 

Finally, it is important to highlight that the period $p_i$ and WCET $c_i^{sens}$, $c_i^{net}$, and $c_i^{ctrl}$ are~known~and considered inputs to our design-time procedure, as we do not want to significantly affect the initial (i.e., non-secured) control deployment. On the other hand, to enforce the tasks' precedence, each control transaction imposes the aforementioned constraints between the offsets and deadlines used to model the transaction tasks. Yet, the actual values are \textbf{not} assigned a priori, i.e.,~the transaction set is considered incomplete, and our goal is to determine offsets and deadlines for all tasks  that produce a schedulable set of control transactions even when security mechanisms are incorporated.

\subsection{Attack Model}

The considered system architecture is susceptible to network-based attacks, such as MitM attacks, on communication between sensors and controllers. 
The attacker can use actions of the controller to force the plant away from the desired state by injecting false data that differ from actual sensor measurements, consequently affecting the controller's estimation and thus the applied control inputs. To formally capture this, we use the standard attack model from~\cite{ncs_attack_models,pajic_csm17,pajic_tcns17, mo2010false,fawzi_tac14}, where additional term $\mathbf{a}_i[k]$ captures the vector of values injected by the attacker at time $k$ on compromised measurements -- i.e., with MitM attacks,  measurements received by the controller  $\hat{\mathbf{y}}_i[k]$ may differ from the actual sensor measurements $ {\mathbf{y}}_i[k]$. Specifically,
\begin{equation}
\label{eq:attModel}
\hat{\mathbf{y}}_i[k] = \begin{cases}
	\mathbf{y}_i[k], & \text{without MitM attack}\\
	\mathbf{y}_i[k] + \mathbf{a}_i[k], & \text{with MitM attack}
\end{cases}
\end{equation}

Due to attacks, the system evolution would not occur according to the model from \eqref{eq:system}. Therefore, we differentiate system evolutions with and without attacks by adding superscript $a$ to all variables affected by the attacker's influence. For example, we denote the plant's state and outputs when the system is under attack as $\mathbf{x}_i^a[k]$ and $\mathbf{y}_i^a[k]$, respectively. Hence, in the case of attacks, sensor measurements delivered to the controller can be modeled as
\begin{equation}
	\hat{\mathbf{y}}_i^a[k]=\mathbf{y}_i^a[k] + \mathbf{a}_i[k]= \mathbf{C}_i\mathbf{x}^a_i[{k}]+\mathbf{v}^a_i[{k}]+ \mathbf{a}_i[k],
\label{eq:systemA}
\end{equation}

The attack vector $\mathbf{a}_i[k]$ is unknown and can have any value assigned by the attacker. The only constraint is that it may be sparse, depending on the set of compromised information flows from sensors to the controller; specifically, if communication from a sensor to the controller for plant $\mathcal{P}_i$ is not corrupted then the corresponding value in $\mathbf{a}_i[k]$ has to be equal to zero. Any assumptions about the set of compromised sensor flows (e.g.,~the number of the flows) can thus be captured by introducing constraints on the sparsity of the vector. However, unless stated otherwise, to simplify our presentation we focus on the worst-case scenario, where the attacker is able to compromise all sensor flows for the plant, once he/she decides to launch an~attack.

%
%


 With the use of standard  cryptographic mechanisms, such as MACs, integrity of the received sensor data can be guaranteed,
 as we assume that the attacker does not have access to the shared secret keys used to generate the MACs.
 In addition, we assume that one of the attacker's goals is \emph{to remain stealthy}, and thus in time steps when message authentication is used, the attacker cannot inject false data (i.e.,~$\mathbf{a}_i[k]=\mathbf{0}$) or the attack will be detected.\footnote{Note that the attacker, with access to the network, could launch Denial-of-Service attacks that prevent messages, including authenticated ones, from being successfully delivered to the controller. In this work, we do not consider such attacks since they are in general easier to detect in CPS with reliable communication networks.}
Furthermore, we assume that the attacker has unlimited computation power and full knowledge of the system, system architecture and plant models, as well as the time-points when authentication will be utilized. This allows him to plan ahead, and smartly craft false measurements to be injected over the network, such that they do not trigger the deployed detector, while deceiving the controller into pushing the plant away from the desired operating~point.\footnote{Examples of such attacks can be found in~\cite{mo2010false, kwon2014stealthy,jovanov_cdc17,jovanov_arxiv17}.}

Consequently, the attacker's goal is to maximally reduce control performance (i.e.,~QoC) while remaining stealthy -- i.e.,~undetected by the system. Therefore, in addition to not inserting false data packets in time-frames when data authentication is enforced, the injected falsified sensor measurements should not trigger the anomaly/intrusion detection system employed at the~controller.

\section{Defending against Attacks with Intermittent Data Authentication}
\label{sec:intermittent}
Enforcing data integrity for every communicated measurement packet may be infeasible due to additional computation costs associated with signing and verifying 
MACs, as well as additional bandwidth required to transmit them.
%
%
%
%
%
%
%
For example, consider three sensing tasks that are being executed on the same ECU, $\{T^{sens}_1(2,10),T^{sens}_2(2,10),T^{sens}_3(5,20)\}$,\footnote{To simplify our notation, when a task $T$ is represented as $T(c, p)$ it is assumed that its offset is equal to zero and relative deadline is equal to the period $p$.}
and let us assume that the security-induced computation overhead to sign measurements with a MAC 
is $2$ time units. 
As shown in \figref{motivationFig}(left), the new task set $\{T^{sens}_1(4,10),T^{sens}_2(4,10),T^{sens}_3(7,20)\}$ is infeasible; thus, even if the network can deal with the additional communication overhead, the 
transmitting ECU cannot authenticate (i.e., sign) every~message.

\begin{figure*}[!t]
    \centering
    \includegraphics[width=\linewidth]{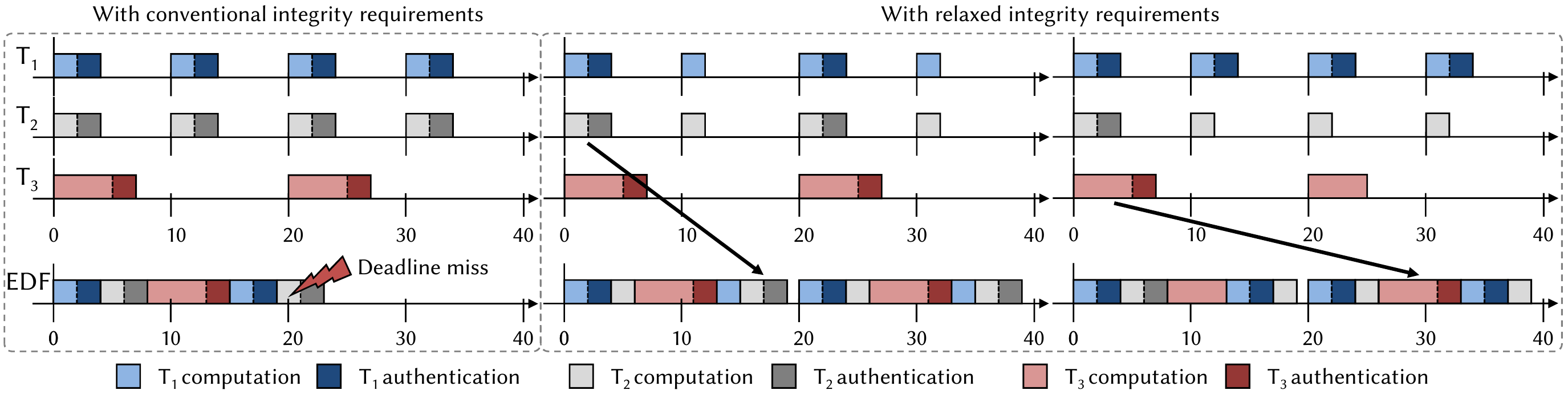}
    \caption{Task set $T^{sens}_1(2,10)$, $T^{sens}_2(2,10)$, $T^{sens}_3(5,20)$ is infeasible if overhead of signing sensor measurements is $2$ time units in every sampling period (left). However, if $T^{sens}_1$ and $T^{sens}_2$ are allowed to authenticate every other period, and the initial authentication of $T^{sens}_2$ is deferred until the second period, the task set is schedulable (center).
    On the other hand, if the goal is to maximize QoC guarantees for the first plant 
    by always authenticating $T^{sens}_1$ 
    measurements, authentication rates for
     $T^{sens}_2$ and $T^{sens}_3$ can be reduced by authenticating every fourth and every other period, respectively, while still providing suitable QoC guarantees under attack,
using the QoC degradation curves to guide formal tradeoff analysis (right).}
    \label{fig:motivationFig}
\end{figure*}

On the one hand, a stealthy attack may significantly reduce QoC if the attacker has compromised a certain number of sensor flows (e.g.,~\cite{pajic_tcns17,kwon2014stealthy}). For any specific class of controllers from~\figref{controllerArch}, by injecting false sensor data that result in a skewed state estimation, the attacker deceives the controller to apply inappropriate control inputs that steer the plant away from the operating point.
On the other hand,  in~\cite{jovanov_cdc17,jovanov_arxiv17,jovanov_cdc18}, we show how physical properties of a system can be exploited to relax integrity requirements for secure control of CPS. The idea is that the state estimation errors due to attacks have to increase slowly to avoid attack detection by the deployed physics-based detector from \figref{controllerArch}. In addition, since each plant has its own dominant time-constant, which can be obtained by the plant model $\mathcal{P}_i$, in the presence of a stealthy attack, QoC can be significantly degraded only after some time has elapsed after the attack is launched.

QoC degradation under attack occurs due to errors in state estimation caused by the false-data injected at time-points when authentication is not used. Hence, for any data authentication policy, which can be captured as time-points where MACs are used (i.e.,~times $k$ where $\mathbf{a}_i[k]=\mathbf{0}$), system performance under stealthy attacks can be evaluated by computing reachable regions of the state estimation error caused by the false data. Specifically, due to stealthy false-data injection attacks, the reachable regions $\mathcal{R}[k]$ and $\mathcal{R}$ of the state estimation error can be defined as~\cite{jovanov_cdc17,jovanov_arxiv17,jovanov_cdc18}
\begin{equation*}
\label{eqn:Rk}
\mathcal{R}[k]=\left\{ \begin{array}{c|c} \mathbf{e} \in\mathbb{R}^n &  \left.\begin{array}{c} \mathbf{e}\mathbf{e}^\intercal\preccurlyeq E[\mathbf{e}^a[k]]E[\mathbf{e}^a[k]]^\intercal + \gamma Cov(\mathbf{e}_k^a),
\\~\mathbf{e}^a[k]=\mathbf{e}_k^a(\mathbf{a}_{1..k}),~\mathbf{a}_{1..k}\in\mathcal{A}_k
\end{array}\right.\end{array}\right\} \quad\text{and}\quad \mathcal{R}= \bigcup_{k=0}^\infty \mathcal{R}[k].
\end{equation*}
Here, $\mathcal{R}$ is the global reachable region of the state estimation error, while $\mathcal{A}_k$ denotes the set of all stealthy attacks $\mathbf{a}_{1..k} = \left[\mathbf{a}[1]^\intercal ... \mathbf{a}[k]^\intercal\right]^\intercal$, and $\mathbf{e}_k^a(\mathbf{a}_{1..k})$ is the estimation error evolution due to the attack $\mathbf{a}_{1..k}$.
Note that this general definition allows for the inclusion of additional information, such as the number and location of compromised sensors. Unless otherwise stated, we assume that measurements from all sensors are compromised when authentication is not used. For instance, \figref{sampleRegions} shows the reachable regions of state estimation error due to stealthy attacks over the adaptive cruise control system described in~Sec.~\ref{sec:caseStudy} for the case with and without intermittent authentication.

\begin{figure*}[!t]
    \centering
    \includegraphics[width=\linewidth]{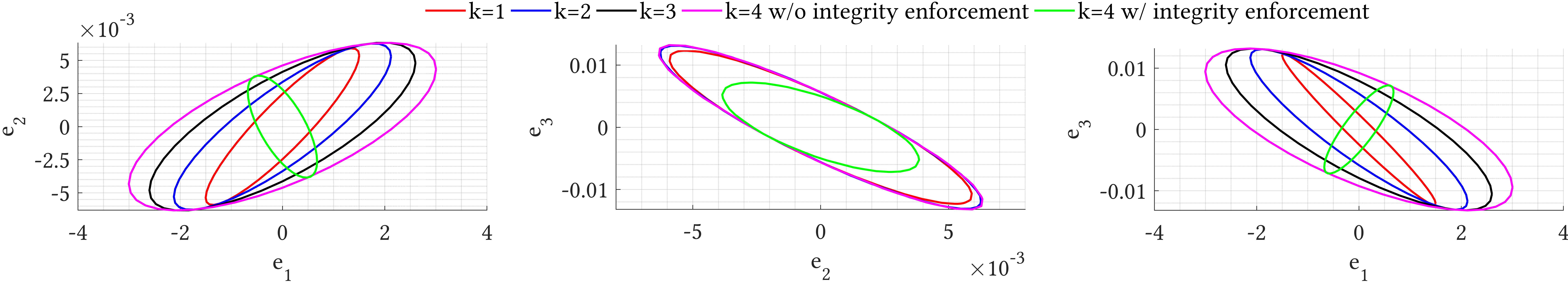}
    \caption{State estimation error evolution due to stealthy attack on distance sensing in an adaptive cruise control system -- projections of the reachable regions in the 3-dimensional state space (distance-speed-acceleration) are shown. Note that the attainable state estimation error is significantly reduced (but not zero) if integrity is enforced over every $4^{th}$ measurement, while the regions grow infinitely without any integrity enforcement.}
    \label{fig:sampleRegions}
\end{figure*}

In \cite{lesi_tecs17}, we introduced 
a \emph{QoC degradation curve} $\mathcal{J}_i(l)$ that, for any linear plant $\mathcal{P}_i$, directly quantifies the dependency between the security-induced computation and bandwidth overhead and the control performance (QoC) under attack, which is reduced due to the estimation errors. Specifically,  $\mathcal{J}_i(l)$ can be used to bound QoC degradation as a function of $l$ -- the maximal time between consecutive uses of MACs in data authentication~policies. This can be formally captured as
\begin{equation*}
\label{eqn:R}
\mathcal{J}_i(l)= supp\{\|\mathbf{e}^a\|_2 ~|~ \mathbf{e}^a\in\mathcal{R}_i^{l}\},~~~~\text{where }~~~~
\mathcal{R}_i^{l} = \cup_{k=0}^\infty \mathcal{R}_i^{l}[k],
\end{equation*}
where $\mathcal{R}_i^{l}[k]$ denotes the reachable region $\mathcal{R}_i[k]$ computed for all data authentication policies with inter-authentication distance of $l$.
%
Such QoC-degradation curves enable the designer to accurately adjust the system's working point by balancing between computational or network resource allocated for security and the returning QoC guarantees under attack, as the predefined QoC requirement can be directly mapped into security-induced overhead and vice versa. 

To illustrate this, let us revisit the example from \figref{motivationFig}, 
and let us assume that for the first two plants, authenticating sensor measurements in every other sampling period ensures the desired QoC level in the presence of attack. \figref{motivationFig}(center) shows that under such conditions, by deferring the initial authentication of $T^{sens}_2$ until the second period, the task set becomes feasible.
Note that, however,  if every fourth measurement for the second plant and every other measurement of the third plant are authenticated, then the measurements for the first plant can continuously be authenticated, as in \figref{motivationFig}(right). QoC degradation curves explicitly capture the dependency between required security-related overhead and control performance, and can be used to decide which scenario is more desired with respect to the overall (for all plants) QoC guarantees. 

\subsection{Cumulative Data Authentication Policies}
In general, depending on the considered plant's dynamics (i.e., matrices $\mathbf{A}_i,\mathbf{B}_i,\mathbf{C}_i$ in~\eqref{eq:system}) it may not be sufficient to intermittently authenticate sensor measurements at one time point. Rather integrity of $f_i$ consecutive measurements should be ensured, with these time-windows 
appearing intermittently during system execution~\cite{jovanov_arxiv17}.\footnote{As shown in~\cite{jovanov_arxiv17}, 
$f=\min(\psi,q^{un}_i)$ with $\psi$ being the observability index of the $(\mathbf{A}_i,\mathbf{C}_i)$ pair and $q^{un}_i$ is the number of unstable eigenvalues of $\mathbf{A}_i$. } 
Implementing such data authentication policies with the use of standard MACs, where every authenticated message is signed with its own MAC added to the message, would require that $f_i$ consecutive communication packets are extended to accommodate MACs.
As the network is commonly a bottleneck in resource-constrained CPS, in this work we propose the use of  \emph{{cumulative} message authentication} where a MAC is computed over several consecutive plant measurements, 
before being attached to the final message from the block; this
significantly reduces the network load by transmitting a MAC for multiple consecutive data points as part of a single message~\cite{streamAuthentication1,nilsson2008efficient}. 
%

Therefore, we introduce the following definitions for cumulative data authentication policies that intermittently or periodically authenticate blocks of messages with sensor measurements.

%



\begin{definition}\label{def:intpolicy}	
An intermittent  cumulative  data authentication policy $\mu_i=(\left\{t_{j}\right\}_{j=0}^\infty, f_i, l_i)$, with 
$t_{{j-1}}<t_{j}$ and $l_i=\sup_{j>0} \left(t_{j} - t_{{j-1}}\right)$, ensures that
$\mathbf{a}_{t_j} = \mathbf{a}_{t_j+1} =...=\mathbf{a}_{t_j+f_i-1} =\mathbf{0}$,  for all $j\geq 0$.
\end{definition}

\begin{definition}\label{def:perpolicy}
	A periodic cumulative data authentication policy $\mu_i(s_i,f_i,l_i)$, 
	where $0\leq s_i\leq l_i-1$,
	ensures that for all $j \geq 0$, 
	$$\mathbf{a}_i[s_i+l_i\cdot j] = \mathbf{a}_i[s_i+1+l_i\cdot j] = ... =\mathbf{a}_i[s_i+f_i-1 + l_i\cdot j] =\mathbf{0}.$$
\end{definition}

\noindent Definition~\ref{def:intpolicy} imposes a maximum time of $l_ip_i$ (i.e., $l_i$ control periods) between the initial authenticated measurements within blocks of $f_i$ consecutive authenticated measurements. On the other hand,~with periodic cumulative authentication policies from Definition~\ref{def:perpolicy}, the time between initial authentications for consecutive blocks is always exactly~$l_ip_i$, and authentication blocks start with the initial offset equal to~$s_ip_i$. 

A control transaction with an intermittent or periodic cumulative authentication policy applied to its tasks (resulting in security-related overheads) is referred to as a {\emph{secure control transaction}}.
%
%
%
For example, consider a secure transaction $\mathcal{T}_i$ from \figref{modelExample}, where a periodic cumulative data authentication policy  $\mu_i(1,2,4)$ is implemented using cumulative MACs. During every four periods, overhead due to MAC signing for sensing task $T_i^{sens}$ is spread over $f_i=2$ jobs, while only one message $M_i^{net}$ and job of $T_i^{ctrl}$ include overhead due to authentication, and only after the last message from the authenticated block is prepared for transmission by~$T_i^{sens}$.

Finally, the use of cumulative authentication introduces delay in verifying data integrity that has to be taken into account when QoC degradation curves are derived. Therefore, in this case QoC degradation curves can be captured as $\mathcal{J}_i(l_i,f_i)$, which are computed from the plant model $\mathcal{P}_i$ using the reachability analysis we introduced in~\cite{jovanov_cdc18}, as illustrated in the upper-right part of  \figref{methodology}.
Since the reachability analysis considers intermittent cumulative authentication policies from Definition~\ref{def:intpolicy}, when used for periodic policies $\mu_i(s_i,f_i,l_i)$, as defined in Definition~\ref{def:perpolicy},  it provides QoC guarantees \textbf{for any value} of $s_i$. For example, the QoC-degradation curves for adaptive cruise control, driveline management and lane keeping controllers, as functions of inter-authentication distance ($l_i$) and authentication block length ($f_i$), are shown in \figref{AllQocCurves}. 
Note that the  adaptive cruise control system requires that at least two consecutive measurements are authenticated (i.e., $f_{ACC}\geq 2$) due to the properties of
the plant's dynamics.

These QoC-degradation  functions $\mathcal{J}_i(l_i,f_i)$ provide the basis for our analysis of tradeoffs between QoC guarantees under attack and the required computational and network resources used for data authentication (i.e., security-related overhead). For each plant $\mathcal{P}_i$, the function $\mathcal{J}_i(l_i,f_i)$ is a non-decreasing function in variable $l_i$. In addition, the minimal required value for $f_i$ can be directly computed from the model of $\mathcal{P}_i$ without significant QoC improvements being obtained by increasing $f_i$. Therefore, the desired QoC requirements (e.g.,~a bound on $\mathcal{J}_i(l_i,f_i)$) can be directly mapped into constraints on the value of $l_i$, the number of non-authenticated communication packets between consecutive block authentications.

\subsubsection{Overview of our Approach}
Our goal is to ensure the desired level of QoC for all controlled plants in resource-constrained CPS, even in the presence of network-based attacks.
As resource constraints prevent continuous authentication of transmitted sensor measurements, we focus on \emph{periodic} cumulative authentication policies, as for such block integrity enforcements are maximally spread apart.
To achieve this, we propose the use of the design-time framework from \figref{methodology}, that directly facilitates tradeoff analysis between the QoC guarantees under attack and security (i.e.,~authentication) overhead for ensuring intermittent integrity of sensor measurements.
For each plant $\mathcal{P}_i$, $i=1,...,N$, the plant model and corresponding QoC curve $\mathcal{J}_i(l_i,f_i)$ are  used to obtain 
constraints on employed periodic cumulative authentication policies; specifically the values for $l_i$ and $f_i$ (but not $s_i$) that result in the desired QoC. In addition, from the platform model and the initial controller specification, regular (i.e., without overheads) and extended (i.e.,  including authentication) WCETs can be obtained, along with the control transaction period $p_i$.

On the other hand, for the task models to be complete and the intermittent authentication policies to be fully defined, it is necessary to derive
feasible (i.e.,~schedulable) tasks' offsets and deadlines, as well as initial authentication offsets ($s_i$) for the cumulative authentication policies.
Consequently, to allow for the execution of secure control transactions with the desired levels of QoC in the presence of attacks, in the rest of the paper we focus on
%
the following scheduling~problems.
\begin{problem}
\label{schedulingProblem}
For a set of secure control transactions 
$\mathcal{T}=\left\{\mathcal{T}_1,...,\mathcal{T}_N\right\}$,
complete the respective task/message sets and deployed periodic cumulative authentication policies, 
such that the obtained secure transaction set $\mathcal{T}$, mapped to available ECUs $\mathcal{E}_1, ...\mathcal{E}_M$, is schedulable under preemptive EDF for ECUs and non-preemptive EDF for the network. 
\end{problem}

\begin{problem}
\label{optimalProblem}

Starting from a schedulable set of secure control transactions
$\mathcal{T}$, obtained from Problem~1, improve the overall QoC guarantees by utilizing remaining resources (ECU time, network bandwidth) with the use of intermittent  cumulative  data authentication~policies.
\end{problem}


%

We consider the use of the EDF scheduler uniformly across ECUs and the network, since EDF is optimal non-idle scheduler for preemptive task scheduling (i.e., on ECUs), while it outperforms rate-monotonic schedulers for realistic loads on non-preemptive networks such as CAN~\cite{relaxingPeriodicityCAN,zuberiShin}. The main challenge in determining unknown parameters (task offsets, deadlines and  extended frame start times) is capturing schedulability conditions for preemptive-EDF on each of the ECUs, as well as non-preemptive-EDF for the shared network. Therefore, in next section, we start by examining the mapping of the control- and security-related platform requirements into a security-aware control transaction model, which will provide a basis for our schedulability analysis and parameter synthesis procedure.

\begin{remark}[Reduction of Control Rate vs. Reduction of Authentication Rate]
The main idea behind 
this work is that with the simultaneous use of physics-based attack detection and cyber-based security mechanisms, such as message authentication, we will be able to provide strong QoC performance guarantees even in resource-constrained CPS, in which it is not possible to protect integrity of every transmitted sensor measurement. An alternative approach to the use of intermittent authentication would be to reduce the control rate to the levels that ensure that every transmitted sensor message can be authenticated. For instance, for our running example from~\figref{motivationFig}, if control task rates are set to $20$, $20$, and $40$ time units respectively (instead of $10$, $10$, and $20$), MACs can protect integrity of every  sensor measurement transmitted over the network. 
However, reducing the control rates (i.e., by increasing control task/sampling periods) results in a reduced control performance in the case without attacks, compared to the initial system that employs the nominal control periods. On the other hand, our goal is to add protection against network-based attacks with strong QoC guarantees in the presence of attacks, without negative effects on control performance (i.e., QoC) when the system is not under attack. With the use of intermittent authentication policies this can be achieved by ensuring schedulability of the main control functionalities (tasks) at the nominal (i.e., initial) periods/rates even when the authentication mechanisms are only intermittently utilized.
\end{remark}

%% file: modeling.tex
\begin{figure}[!t]
  \centering
  \includegraphics[width=.668\linewidth]{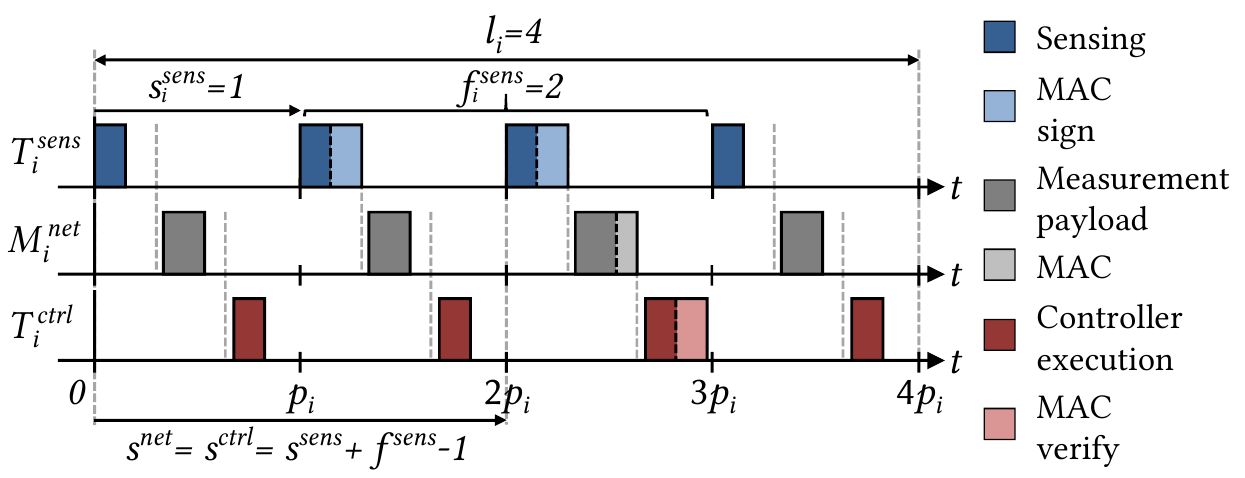}
  \caption{Example of a secure control transaction 
  when the periodic cumulative data authentication policy  $\mu_i(1,2,4)$ is used. 
   Note that only the transmitting task is extended for $f_i$ consecutive invocations to perform cumulative authentication. On the other hand, the network message is extended only once, and accordingly the receiving task performs authentication (i.e., verifies the received MAC) once after receiving that measurement.}
\label{fig:modelExample}
\end{figure}

\section{Modeling Secure Control Transactions}
\label{sec:modeling}

Let us consider the workload imposed by a secure control transaction, such as the one shown in \figref{motivationFig} (center, right).
 Schedulability analysis for such workloads using the standard task model $(WCET,period,deadline)$ is highly pessimistic -- clearly, the task sets from the figures would be rejected; the reason is that the standard task and message models accepting a single WCET parameter coarsely overapproximate the load on the ECUs and the shared network imposed by sparsely added security overhead. 
 Thus, we need a 
 model that captures the variable execution (or transmission) times of such security-aware real-time~tasks. 

 The multi-frame task model \cite{Multiframe} supports tasks that have execution times varying among consecutive invocations (called \emph{frames}) in an arbitrary pattern. However, this model is overly general in that it allows any pattern of frames to be specified, and schedulability analyses for multi-frame tasks often assume that the worst-case alignment of frames is legal --- exactly the scenario we want to avoid.
In our case, it suffices to facilitate two frame sizes, regular and extended, with extended corresponding to executions that include security-related overhead, as well as additional parameters specifying extended frame period and offset; this allows for capturing of periodic cumulative data authentication policies,
as the ones applied to tasks in Fig.~\ref{fig:motivationFig} (center, right). 

Our goal is to develop a methodology for completing a set of transactions on the available shared network and set of ECUs, while taking into account the required level of periodic data integrity guarantees, which are obtained from the predefined QoC under attack requirements. Thus, we assume that non-zero task offsets and constrained deadlines are not known a priori. Instead, the respective task   sets are considered incomplete in the sense that their periods and execution/transmission times are known, but the offsets and deadlines for each of the tasks that produce a schedulable set of transactions are to be determined. Consequently, we model the \emph{security-aware tasks} as $T_i(C_i,p_i,\phi_i,d_i,l_i,f_i,s_i)$, where
%
%
\begin{itemize}
    \item $C_i=[c_i^{reg},c_i^{ext}]$ is a WCET array for two frame types, regular and extended, respectively -- $c_i^{reg}$ is equal to $c^{sens}, c^{net}$ or $c^{ctrl}$ for  $T_i^{sens}$, $M_i^{net}$ and $T_i^{ctrl}$, respectively;
    \item $p_i$ is the period at which jobs are released, $\phi_i$ is the release offset, and $d_i$ is the task's deadline relative to its activation;
    \item $l_i$ is the distance (i.e., number of control periods) between consecutive authentication blocks;
    \item $f_i$ captures the length of the authentication block -- i.e., 
    the number of authenticated frames within one authentication period (i.e., within every interval of length~$l_ip_i$);
    \item $s_i$ is the
     initial authentication 
     offset
    -- i.e., the integer multiple of periods by which the initial authentication is deferred.
\end{itemize}
%
Note that the task offset consists of two components: $\phi_i$ and $s_ip_i$;
 $\phi_i$ is required to encode precedence constraints and applies to all jobs of the considered $i^\text{th}$ task. On the other hand, $s_ip_i$ determines the additional offset of only extended frames, which 
provides a degree of freedom during scheduling to avoid extended frame alignment scenarios emphasized in the motivating~example (Fig.~\ref{fig:motivationFig} (left)).


For tasks in any secure control transaction $\mathcal{T}_i$, some of the above parameters (i.e., $s_i,f_i,l_i$)  directly follow from the employed authentication policy $\mu_i(s_i,f_i,l_i)$, as
illustrated in~\figref{modelExample} for one example transaction.
First, $l_i^{sens}=l_i^{net}=l_i^{ctrl}=l_i$, since the authentication period is the same for both tasks and the communication message. 
In addition, $f_i^{sens}=f_i$, as $T_i^{sens}$ task computes a cumulative MAC over a block of $f_i$ consecutive measurements, before attaching the MAC to the last message from the block.
Also, $f_i^{ctrl}=1$ 
since 
$T_i^{ctrl}$ task verifies (i.e., authenticates) a block of consecutive measurements only once when it receives the cumulative MAC, prepared by $T_i^{sens}$ and delivered by $M_i^{net}$.
Thus, 
it also holds that $f_i^{net}=1$. 

Similarly, initial authentication offsets depend on the authentication policy used. First, $0\leq s_i^{sens}\leq l_i-f_i$ since the first computation of cumulative MAC within a block must be done early enough to allow for execution of $f_i$ consecutive extended frames within $l_i$ periods of $T_i^{sens}$. Additionally, the initial extended frames of the message $M_i^{net}$ and control task $T_i^{ctrl}$ have constrained start times as $s_i^{ctrl}=s_i^{net}=s_i^{sens}+f_i^{sens}-1$, as $T_i^{sens}$ 
task 
computes cumulative MAC over $f_i^{sens}$ periods, followed by an authenticated transmission and an authenticating control task, 
as shown in~\figref{modelExample}.

Problem~1 can now be reformulated around synthesis of feasible deadlines ($d_i^{sens}, d_i^{net}, d_i^{ctrl}$), offsets  ($\phi_i^{sens}, \phi_i^{net}, \phi_i^{ctrl}$) and initial authentication offsets ($s_i^{sens}, s_i^{net}, s_i^{ctrl}$) for all secure control transactions $\mathcal{T}_i$, $i=1,...,N$, such that the precedence constraints from~\eqref{eq:c1}-\eqref{eq:c3} are satisfied, and for which the obtained complete transaction set $\mathcal{T}$ is schedulable under preemptive EDF for ECUs and non-preemptive EDF for the network.
Thus, the following section starts by deriving schedulability conditions for the presented task model under preemptive and non-preemptive EDF~scheduling.

%% file: scheduling.tex

\section{Schedulability Analysis for Secure Control Transactions}
\label{sec:schedulabilityAnalysis}

\subsection{Schedulability of Security-Aware Tasks}
We consider a schedulability condition for the sensing and control tasks based on the \emph{processor demand criterion}~\cite{Baruah1990}. 
Note that the condition from~\cite{lesi_tecs17} cannot be used as it does not support the use of cumulative periodic authentication on sensing tasks, as well as  general offset and deadline values for tasks and messages in secure control transactions. 
On the other hand, necessary and sufficient schedulability conditions for the general task model (i.e., with non-zero offsets and deadlines differing from periods) under the preemptive EDF scheduler are formulated in~\cite{Baruah1990,ButtazzoBook}, starting from the~following.
%


\begin{definition}[\cite{Baruah1990}]
\label{def:dbfDef}
    The demand function $df_i$ of a standard task $T_i(c_i,p_i,\phi_i, d_i)$ on  interval $[t_1,t_2]$ is
$       ~~df_i(t_1,t_2) = \sum\limits_{\substack{\alpha_{i,j} \geq t_1,~\delta_{i,j} \leq t_2}}c_i$,
    where $c_i$ is the WCET of the $i^{\text{th}}$ task, while $\alpha_{i,j}$ represents the time of the  $j^{\text{th}}$ job arrival, and $\delta_{i,j}$ its respective deadline.
\end{definition}
\begin{theorem}[\cite{Baruah1990}]
\label{thm:feasibilityCond}
    A task set $\{ T_1(c_1,p_1,\phi_1,d_1), T_2(c_2,p_2,\phi_2,d_2)$,..., $T_N(c_N,p_N,\phi_N,d_N) \}$ is schedulable by preemptive EDF if and only~if
$        \sum_{i=1}^Ndf_i(t_1, t_2)\leq t_2-t_1,$ for all $t_1, t_2$ such that $t_1< t_2$.
\end{theorem}

Since, by definition, the demand function is piecewise constant with magnitude increasing in steps at time instants of job deadlines, the condition in \thmref{feasibilityCond} can be evaluated over a discrete and bounded time testing set. Formally, it is necessary to test the processor demand condition for all $t_{k_1}< t_{k_2} \leq t^{max}$ such that
\begin{equation}\label{eq:timeTestingSets}
    \begin{split}
      t_{k_1} \in TS_{arr} = &\bigcup_{i=1}^{N}\{ t | t=\phi_i+k_1 p_i, k_1\in\mathbb{N}_0, t\leq t^{max}\},\\
      t_{k_2} \in TS_{dead} = &\bigcup_{i=1}^{N}\{ t | t=d_i+k_2 p_i, k_2\in\mathbb{N}_0, t\leq t^{max}\},
    \end{split}
\end{equation}
where $t^{max}=\max_{i}\phi_i+\max_{i}d_i+ 2\cdot lcm\{ p_1,...,p_N \}$ is the maximal time up to which the CPU demand has to be tested to ensure correctness of analysis~\cite{LeungMerrill}, and $lcm$ is the least common~multiple.

We use this schedulability condition for schedulability analysis of security-aware 
$T_i^{sens}$ and $T_i^{ctrl}$ tasks -- to simplify notation, we omit superscripts and denote the tasks as $T_i$ where possible. 
To evaluate the demand function on interval $[t_{k_1},t_{k_2})$, we compute the number of regular and extended frames released at or after $t_{k_1}$, that have deadlines at or before $t_{k_2}$ as
\begin{equation}\label{eq:etaNormal}
     \eta_i^{r\&e}(t_{k_1},t_{k_2}) = max\left\{ 0,\floor*{\frac{t_{k_2}-\phi_i-d_i}{p_i}} - max \left\{ 0,\ceil*{\frac{t_{k_1}-\phi_i}{p_i}} \right\} +1 \right\}.
\end{equation}
Similarly, extended frames in this interval can be counted as
\begin{equation}\label{eq:etaExtended}
    \begin{split}
       \eta_i^{ext}(t_{k_1},t_{k_2})= \sum_{m=0}^{f_i-1}\:\: &max \left\{ 0,\floor*{\frac{t_{k_2}-(s_i+m)p_i-\phi_i-d_i}{l_ip_i}}-\right.\\
       &max \left. \left\{0,\ceil*{\frac{t_{k_1}-(s_i+m)p_i-\phi_i}{l_ip_i}} \right\} +1  \right\}.
    \end{split}
\end{equation}
Here, the appropriate values for $f_i$ should be used -- i.e.,~$f_i^{ctrl}=1$ for $T_i^{ctrl}$ and $f_i^{sens}=f_i$ for $T_i^{sens}$.

The demand function for a single task can now be posed as the total processor demand of regular and extended frames~as
\begin{equation}\label{eq:dfK1K2}
    df_i(t_{k_1},t_{k_2}) = c_i^{reg}\eta_i^{r\&p}(t_{k_1},t_{k_2})+\Delta c_i\eta_i^{ext}(t_{k_1},t_{k_2}),
\end{equation}
where $\Delta c_i = c_i^{ext}-c_i^{reg}$. We can thus formulate the necessary and sufficient schedulability condition as: $\forall t_{k_1}\in TS_{arr},\:\forall t_{k_2}\in TS_{dead}$
\begin{equation}\label{eq:sumDemandCondition}
           \sum\limits_{i=1}^{N} df_i(t_{k_1},t_{k_2}) \leq t_{k_2}-t_{k_1},\\~~
            \hbox{if }~ t_{k_1}<t_{k_2}.
\end{equation}

\subsection{Schedulability of Security-Aware Messages}
To analyze schedulability of security-aware network messages (i.e., with periodic cumulative authentication), 
we start from 
the following theorem that provides a necessary and sufficient schedulability condition for \emph{sporadic} real-time messages under non-preemptive~EDF.

\begin{theorem}[\cite{nprEDF-CAN}]
\label{thm:nonpreemptiveSporadicTheorem}
Consider a set of real-time messages $M_i(c_i,p_i,d_i)$, $1\leq i \leq N$, where $p_i$ is the minimum message inter-arrival time. 
The message set is schedulable under non-preemptive EDF over a network shared with non real-time messages with maximum transmission time $c_{max}^{NRT}$~if and only if $\sum_{i=1}^{N}\frac{c_i}{p_i}\leq 1$ and
\begin{equation}
\label{eq:NP_schcond}
\sum_{i=1}^{N}max\left\{0,\floor*{\frac{t-d_i}{p_i}}+1\right\}c_i + c_m \leq t_k, \forall t_k \in TS,
\end{equation}
where $TS=\bigcup\limits_{i=1}^{N}\left\{d_i+jp_i|j=0,...,\floor*{\frac{t_{max}-d_i}{p_i}}\right\}$,\\
$t_{max}=\max\left\{ d_1,...,d_N,\left(c_m+\sum_{i=1}^{N}\left(1-\frac{d_i}{p_i} \right)c_i\right) / (1-U_\mathcal{M}) \right\}$, and
$c_m=max\{c_{max}^{NRT}, \max_{i=1}^{N}c_i\}$.
\end{theorem}

\begin{figure}[!t]
  \centering
  \includegraphics[width=0.53\linewidth]{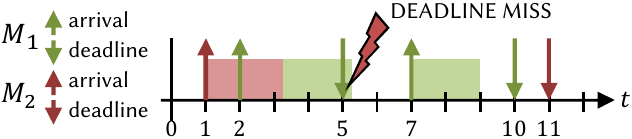}
  \caption{Example message set $M_1(\phi_1=2,c_1=2,p_1=5,d_1=3), M_2(\phi_2=1,c_2=2.1,p_2=10,d_2=10)$ --- although the schedulability test for nonpreemptive messages with offsets from~\cite{zuberiShin} is satisfied, $M_1$ misses its deadline at $t=5$ due to an earlier release of message~$M_2$.}
\label{fig:falseExample}
\end{figure}

To the best of our knowledge, there does not exist an efficient method 
to test schedulability for strictly periodic asynchronous messages under non-preemptive EDF. The conditions from~\cite{zuberiShin} extend \thmref{nonpreemptiveSporadicTheorem} for messages with offsets in order to support transaction scheduling. The resulting theorem from~\cite{zuberiShin} replaces every appearance of relative deadline $d_i$ in \thmref{nonpreemptiveSporadicTheorem} with absolute deadline $d_i+\phi_i$ to account for offsets. In our case, using this theorem would be pessimistic since the conditions derived for sporadic messages cannot be adjusted for multi-frame messages. Also, examples as in \figref{falseExample} show that the schedulability condition from~\cite{zuberiShin} does not always~hold. 


On the other hand, a utilization-based test for non-preemptive EDF is derived in~\cite{sanjoyNPRedf}.
As our goal is to determine a set of offsets and deadlines that yields a schedulable set of secure transactions, this test 
cannot be used as it condenses all task properties into a single measure.
Still, by following the reasoning presented therein, we formulate  the following sufficient schedulability condition. 

\begin{theorem}
\label{thm:NPRfeasibilityCond}
    A message set $\{ M_1(c_1,p_1,\phi_1,d_1), M_2(c_2,p_2,\phi_2,d_2)$, ..., $M_N(c_N,p_N,\phi_N,d_N) \}$ is nonpreemptively schedulable by EDF if $\sum_{i}df_i(t_1, t_2)\leq t_2-t_1-c_{max},$ for all $t_1, t_2$ such that $t_1< t_2$, where $c_{max}=\max_{i} c_i$ is the longest of transmission times of all $N$ messages.
\end{theorem}
\begin{proof}
    Suppose that the theorem's demand-based condition is satisfied for all $t_1, t_2$, and that there is a deadline miss at some instant $t_2^*=t_{dm}$. 
Let $t_1^* \leq t_{dm}$ be the closest to $t_{dm}$ instant such that the network is busy transmitting only those messages with deadlines $\leq t_{dm}$. Then, right before $t_1^*$, the network may be idle or a message with deadline $\geq t_{dm}$ is being transmitted.

    In the case when the network is idle right before $t_1^*$, then the total network demand imposed by all messages eligible to be transmitted during $[t_1^*,t_2^*]$ is  $\sum_{i}df_i(t_1^*, t_2^*)$, by the definition of the demand function, and since there is a deadline miss at $t_2^*$, the demand must be greater than the network time available, i.e., $\sum_{i}df_i(t_1^*, t_2^*) > t_2^*-t_1^*$. This contradicts the theorem statement.

    In the case when the network is transmitting a message with deadline $\geq t_{dm}$, then the worst case network demand of all messages eligible to be transmitted during $[t_1^*,t_2^*]$ is $\sum_{i}df_i(t_1^*, t_2^*)+c_{max}$. Since there is a deadline miss at $t_2^*$, the demand must be greater than the available network time, i.e., $\sum_{i}df_i(t_1^*, t_2^*) > t_2^*-t_1^*-c_{max}$, which contradicts the theorem, and thus concludes the~proof.
\end{proof}

The intuition behind this theorem can be supported by the claim that non-preemptive EDF schedules by time $t^*+c_{max}$ at least as much work imposed by a set of tasks as preemptive EDF schedules by $t^*$~\cite{sanjoyNPRedf}.
 In this case, 
 the total network demand by a security-aware message can be expressed as in~\eqref{eq:dfK1K2}, with $f_i^{net}=1$  used for extended transmissions in~\eqref{eq:etaExtended}. 
%
In addition, 
the time testing sets remain the same as in~\eqref{eq:timeTestingSets}. As we demonstrate on examples in Section~\ref{sec:generalEvaluation}, this condition is less conservative in cases when message transmission times are significantly shorter than their respective periods. 
We then show in Section~\ref{sec:caseStudy} that this is commonly true in practical systems.

\begin{remark}[Accounting for Jitter]
To understand how realistic implementation phenomena such as jitter affect the presented analysis, we consider their effects on task and message scheduling. 
In the case of task-level jitter, existing approaches to jitter accounting can be applied~\cite{StankovicEDF}. In essence, if a task experiences jitter $j_i$, the inter-arrival spacing may be shorter than $p_i$. From the worst-case schedulability standpoint, this scenario pertains to the arrival pattern where all tasks arrive such that they must complete execution by the relative deadline $d_i-j_i$, rather than by $d_i$ time units. Shortening the permissible deadline by the worst-case jitter can be easily included in the demand-based condition~\eqref{eq:sumDemandCondition}. This does not affect the complexity of the MILP implementation of the parameter synthesis problem, as worst-case jitter figures as a set of known constant parameters.
For message scheduling, in most cases we do not need to use this approach, as the $c_{max}$ term introduced in the non-preemptive schedulability conditions to account for the worst-case blocking any message may experience upon arrival, is rarely needed in its entirety; 
 this holds since worst-case blocking will rarely occur. 
 This conservativeness effectively captures jitter,
as jitter levels are highly unlikely to exceed message transmission times in any practical network realization. 
\end{remark}

%% file: MILP.tex

\section{Synthesis of Schedulable Secure Control Transactions} 
\label{sec:MILP}

The schedulability conditions from Section~\ref{sec:schedulabilityAnalysis}, along with the task-precedence constraints from Section~\ref{sec:CT},  can be used to formulate a parameter synthesis problem that produces a feasible set of task deadlines, offsets, and initial authentication offsets.
 However, non-linearity of functions counting the number of task invocations and message transmissions~\eqref{eq:etaNormal} and ~\eqref{eq:etaExtended} precludes efficient search of the parameter space. 
 Thus, in this section we map the demand-based schedulability conditions into a set of linear constraints, and formulate  a mixed-integer linear program~(MILP) to synthesize  task and message parameters that result in a schedulable set of secure control transactions. 
Since the schedulability conditions for preemptive and non-preemptive EDF differ only in the constant term $c_{max}$ on the right side of the demand constraints from Theorems~\ref{thm:feasibilityCond} and~\ref{thm:NPRfeasibilityCond}, 
in this section we may omit superscripts $sens$, $ctrl$, and $net$ for specific variables, where no confusion  about 
the parameters~arises.

Consider the workload of a sensing task $T_i^{sens}$ that also incorporates 
cumulative periodical authentications. 
Let binary variables $a_{k,j,m}^i$ for 
$T_i^{sens}$ indicate that the absolute deadline of the $m^\text{th}$ extended frame of the $j^\text{th}$ block of cumulative authentications is at or earlier than a time-testing instant $t_k$. This 
can be specified~as
\begin{equation}\label{eq:aVariables}
    a_{k,j,m}^i = 1  \Leftrightarrow t_k\geq (s_i+m)p_i+\phi_i+d_i+(j-1)l_ip_i,
\end{equation}
\begin{equation}\label{eq:aVariablesIndices}
    \begin{aligned}
        1\leq i\leq N,\qquad &1\leq k \leq |TS_{arr}|+|TS_{dead}|,\qquad
        1\leq j\leq \floor*{\frac{t^{max}}{l_ip_i}},\qquad &0\leq m\leq f_i-1,
    \end{aligned}
\end{equation}
where $TS_{arr}$ and $TS_{dead}$ are defined in~\eqref{eq:timeTestingSets}. Note that control tasks $T_i^{ctrl}$ and  messages $M_i^{net}$ are supported by simply removing the authentication iterator $m$ (since $f_i^{ctrl}=f_i^{net}=1$). A similar relation can be established for regular frames, where binary variables $b_{k,h}^i$ indicate that the $h^\text{th}$ regular frame of the $i^\text{th}$ sensing task is due by the $k^\text{th}$ time testing instant $t_k$. This can be captured~by
\begin{equation}\label{eq:bVariables}
    b_{k,h}^i = 1  \Leftrightarrow t_k\geq \phi_i+d_i+(h-1)p_i, \quad
        1\leq i\leq N, 1\leq k \leq |TS_{arr}|+|TS_{dead}|,1\leq h\leq \floor*{\frac{t^{max}}{p_i}}.\\
\end{equation}
Identical constraint can be written for control tasks $T_i^{ctrl}$ and  messages $M_i^{net}$. 
These variables enable us to concisely specify the number of respective jobs from \eqref{eq:etaNormal}~and~\eqref{eq:etaExtended} respectively~as
\begin{align}
    \eta_i^{r\&e}(t_{k_1},t_{k_2}) &= \sum_{j=1}^{\frac{t^{max}}{p_i}}\left(b_{k_2,h}^i-b_{k_1,h}^i\right),
    \label{eq:n1}\\
    \eta_i^{ext}(t_{k_1},t_{k_2}) &= \sum_{m=0}^{f_i-1}\sum_{j=1}^{\frac{t^{max}}{l_ip_i}}\left(a_{k_2,j,m}^i-a_{k_1,j,m}^i\right).
    \label{eq:n2}
\end{align}
Hence, a task's processor demand can be cast as a linear function of variables $a_{k,j,m}^i$ and~$b_{k,h}^i$ when~\eqref{eq:n1},~\eqref{eq:n2} are instantiated in~\eqref{eq:dfK1K2}.
Note that since network and ECUs may not have the same hyperperiod, $t^{max}$ should be computed independently for each ECU.

%
%

Note that, since task offsets and deadlines are variables, the time testing instants are also variables, as defined in~\eqref{eq:timeTestingSets}. Therefore, 
we need to ensure that we only consider the schedulability constraints from  Theorems~\ref{thm:feasibilityCond} and~\ref{thm:NPRfeasibilityCond} for  $k_1$ and $k_2$ such that $t_{k_1} < t_{k_2}$.
 This is achieved with a set of constraint-enabling variables $e_{k_1,k_2}$ such that
\begin{equation}\label{eq:demandEnableRelation}
    e_{k_1,k_2}=1 \Rightarrow \sum\limits_{i=1}^{N} df_i(t_{k_1},t_{k_2}) \leq t_{k_2}-t_{k_1},
\end{equation}
for preemptive EDF, where $e_{k_1,k_2}$ relates to the time testing instants as
\begin{equation}\label{eq:demandEnableVariablesRelation}
    e_{k_1,k_2}=1 \Leftrightarrow t_{k_2} > t_{k_1}.
\end{equation}
%
%
In addition,  
the right side of~\eqref{eq:demandEnableRelation}
should be decremented by $c_{max}^{net}$ when considering message scheduling, due to the scheduling non-preemptivity 
(\thmref{NPRfeasibilityCond}).

Finally, to impose a bounded end-to-end delay,  constraints that relate deadlines of tasks in a transaction can be specified as
\begin{equation}\label{eq:boundedDeadlines}
    d_i^{sens}+d_i^{net}+d_i^{ctrl}=p_i, \qquad 1 \leq i\leq N.
\end{equation}

\begin{remark}[Handling of Indicator Constraints]
    While the processor demand conditions can be directly implemented within an MILP, constraints~\eqref{eq:aVariables},~\eqref{eq:bVariables}, \eqref{eq:demandEnableRelation}, and \eqref{eq:demandEnableVariablesRelation} cannot be directly specified as such in some MILP solvers. Those constraints can be linearized by using the "Big M" method for handling indicator constraints~\cite{bigM}. In the case of~\eqref{eq:aVariables} and ~\eqref{eq:bVariables}, we can write
    \begin{equation}\label{eq:aConstraintsBigM1}
        -t_k+\phi_i+d_i+Ma_{k,j,m}^i \leq M-[ s_i+m+(j-1)l_i ]p_i,
    \end{equation}
    \begin{equation}\label{eq:aConstraintsBigM2}
        t_k-\phi_i-d_i-Ma_{k,j,m}^i < [ s_i+m+(j-1)l_i ]p_i,
    \end{equation}
    \begin{equation}\label{eq:bConstraintsBigM1}
        -t_k+\phi_i+d_i+Mb_{k,h}^i \leq M-(h-1)p_i,
    \end{equation}
    \begin{equation}\label{eq:bConstraintsBigM2}
        t_k-\phi_i-d_i-Mb_{k,h}^i < (h-1)p_i,
    \end{equation}
    where $M$ is a large constant. Similarly,~\eqref{eq:demandEnableRelation} and~\eqref{eq:demandEnableVariablesRelation} can be cast as linear constraints by enforcing
    \begin{equation}\label{eq:dfConstraintsBigM}
        M(e_{k_1,k_2}-1)+\sum\limits_{i=1}^{N} df_i(t_{k_1},t_{k_2}) \leq t_{k_2}-t_{k_1},
    \end{equation}
    \begin{equation}\label{eq:eConstraintsBigM1}
        t_{k_2}-t_{k_1} > M(e_{k_1,k_2}-1),
    \end{equation}
    \begin{equation}\label{eq:eConstraintsBigM2}
        t_{k_2}-t_{k_1} <   Me_{k_1,k_2}.
    \end{equation}
\end{remark}
\begin{remark}[Handling of Strict Inequalities]
    Most MILP solvers do not allow specification of strict inequalities. Constraints~\eqref{eq:aConstraintsBigM2} and~\eqref{eq:bConstraintsBigM2} can be converted into non-strict inequalities by adding a small $\epsilon>0$ to every $t_k$. Furthermore, \eqref{eq:eConstraintsBigM1}~can be directly converted into non-strict inequalities, while~\eqref{eq:eConstraintsBigM2} requires addition of a small $\epsilon>0$ on the left-hand side. Note that this may allow the time testing instants to meet during the solving process i.e., $t_{k_1}=t_{k_2}$ is possible for some pair $(k_1,k_2)$. This does not affect correctness of the formulation, but can only introduce redundant trivial demand constraints (i.e., over the interval of zero length). However, this does create an undesirable corner case. Despite the lack of an objective (recall that we are only interested in finding a feasible solution if such exists), solvers tend to minimize variables, and may thus choose to zero all deadlines. This corner case is formally allowed if a time testing instant corresponding to a deadline of a task can coincide with its arrival. Since the demand constraint is satisfied (the processor demand over the interval of length zero is equal to the supply over the same interval), this modeling anomaly requires lower-bounding deadlines of each of the tasks. Simply, $d_i \geq 1$,  for all $i$ suffices.

Additionally, introducing $\epsilon$ to handle strict inequalities may affect the choice of value for $M$. Specifically, the values for $M$ and $\epsilon$ must be selected such that 
 no negative effects occur with the use of "big-M" methods
 due to finite precision implementation of the employed MILP solver --- that no constraint become active due to finite values for $M$. Thus, we set these values such that it holds~that
    \begin{equation*}\label{mEpsilonCond}
        M\delta_{int}+\delta_{constr} < \epsilon < 1-M\delta_{int}-\delta_{constr},
    \end{equation*}
    where $\delta_{int}$ is the integer feasibility tolerance and $\delta_{constr}$ is the constraint satisfiability tolerance of the employed MILP solver. Moreover, $M$ must be sufficiently large to ensure constraint satisfiability is not compromised for large $t_k$-s from the set $TS$.
\end{remark}

The aforementioned constraints form a MILP formulation whose variables  are the deadlines ($d_i^{sens}, d_i^{net}, d_i^{ctrl}$), offsets  ($\phi_i^{sens}, \phi_i^{net}, \phi_i^{ctrl}$) and initial authentication offsets ($s_i^{sens}, s_i^{net}, s_i^{ctrl}$), as well as the introduced binary variables,
but without an objective specification. 
If the feasible set of the problem is non-empty, our transaction set becomes complete and guaranteed schedulable. This approach, however, may be impractical for realistic scenarios. For example, a unified MILP for the case study presented in Section~\ref{sec:caseStudy} features over 10 million variables and 100 million constraints. Therefore, in the rest of the section, we discuss methods for complexity reduction that we apply towards tackling realistic~problems.
%
%
%
%
\subsection{Complexity Reduction}
\label{sec:complexity}
To reduce the number of used variables and constraints, we first
consider the time testing sets in~\eqref{eq:timeTestingSets} for preemptive EDF. 
For a large number of arrival-deadline pairs $(t_{k_1},t_{k_2})$, defining a variable indicating their ordering as in~\eqref{eq:demandEnableVariablesRelation} is not necessary, and thus the corresponding demand constraints can be omitted. 
For example, 
arrival time of any single job may never exceed the deadline of that, or any subsequent invocations of the task. Also, the deadline of a specific task invocation always occurs after the arrival of that or any earlier task invocations. Formally, since $e_{k_1,k_2}=0, \forall i, \forall k_2 \geq k_1 \text{~such that~} \phi_i+k_1p_i \geq d_i+k_2p_i$, and $e_{k_1,k_2}=1, \forall i,\forall k_2 \geq k_1 \text{~such that~}  \phi_i+k_1p_i < d_i+k_2p_i$, constraints~\eqref{eq:dfConstraintsBigM}--\eqref{eq:eConstraintsBigM2} can be omitted.
Similar relations can be drawn pairwise for every two tasks, given specific temporal parameters. 
This approach greatly reduces the number of used variables and constraints, especially for large~hyperperiods.

A similar reasoning can be applied to variables $a_{k,j,m}^i$ that control the peak-frame timing. Given specific temporal parameters of tasks, it is not necessary to encode appearance of the $j^\text{th}$ authentication block (i.e., $j^\text{th}$ sequence of $m$ consecutive peak frames) for all instants in the time testing set, as suggested by the general definitions given in ~\eqref{eq:aVariables}--\eqref{eq:aVariablesIndices}. This is true since we only seek to find a schedulable solution, which implies that the $j^\text{th}$ authentication block must occur during the interval $\left[(j-1)l_ip_i,jl_ip_i\right]$, outside of which the value of $a_{k,j,m}^i$ is fixed and fully determined by tasks' temporal parameters. Formally,
$$(\forall i,j,k,m) (t_k > j l_i p_i \Rightarrow a_{k,j,m}^{i}=0 \:and\: t_k < (j-1) l_i p_i \Rightarrow a_{k,j,m}^{i}=1).$$
Similar holds for normal frames that must be scheduled within their respective periods:
$$(\forall i,k,h) (t_k > h p_i \Rightarrow b_{k,h}^{i}=0 \:and\: t_k < (h-1) p_i \Rightarrow b_{k,h}^{i}=1),$$
and thus the majority of constraints~\eqref{eq:bVariables} 
and corresponding variables $b_{k,h}^i$ that control normal frame timing can be eliminated. 
By enforcing these rules during problem encoding, the number of variables and constraints required to encode a realistic problem vastly reduces.


\subsection{MILP Decomposition}
\label{sec:decomposition}
Even with the discussed reductions in number of variables and constraints, the presented MILPs may remain relatively complex for very large transaction sets. For these scenarios, we propose a decomposition approach that formulates the synthesis of schedulable secure control transactions as a sequence of MILPs, rather than a single program, since the  schedulability tests from  Section~\ref{sec:schedulabilityAnalysis} can be decoupled between the ECUs and network. However, as we consider a parameter synthesis problem, rather than just a schedulability test, this decomposition is nontrivial --  schedulable task parameters obtained for one part of the system do not guarantee feasibility of the remaining parts. In fact, the decomposition approach 
directly depends on the system architecture and its~implementation.

\subsubsection{Synchronous Sensing Platform Model}
\label{subsec:syncSens}
The most commonly adopted platform model in offset-based scheduling of control transactions (as in~\cite{relaxingPeriodicityCAN}) assumes that all sensing tasks are initially released at the same time  (i.e.,~$\forall i, \phi_i^{sens}=0$ and $t_0=n\cdot p_i$ for some $n$ in \figref{endToEndTiming}). In this case, in the first stage, we can run the ECUs' parameter synthesis MILP. In this case, our objective could ensure that sensing tasks are scheduled as early as possible (minimized deadlines) while the opposite is desired for receiving tasks,~i.e., they should execute as late as possible during their respective periods (maximized offsets), in order to ensure that the least conservative timing constraints are imposed on network messages (\figref{endToEndTiming}).
Trying to minimize all $d_i^{sens}$ and maximize all $\phi_i^{ctrl}$ effectively results in multivariate optimization that we solve by associating weights with each of the objectives (i.e., using blended objective).
In the second stage, the network parameter synthesis MILP is formulated as a feasibility problem (without objective) searching for message offsets and deadlines that yield in a feasible transaction set. 

Alternatively, in the first stage, we can run the network parameter synthesis MILP with the objective to maximize message offsets $\phi_i^{net}$ (which `leaves' time for transmitting tasks) and minimize deadlines $d_i^{net}$ (which `leaves' time for receiving tasks). However, these objectives are conflicting, and since they have to be specified as a single blended objective function, heuristics can be used to adjust weights of individual message offsets and deadlines according to the execution times of sensing and control tasks (i.e.,~if the sensing task's WCET is longer than the control task's WCET, the message should be delayed more towards the end of the period). In the second stage, the ECUs' parameter synthesis MILP is formulated as a feasibility problem. However, there exist scenarios where this model is not the most accurate one; for instance, an ECU attached to multiple sensors may not necessarily have the capability to sample them instantaneously.

Consequently, our approach is that  in the first stage we execute the MILP formulation with lower complexity, that is better suited for this architecture, since that would reduce the time cost of reconfiguring task sets in the case that the MILP solver initially returns no solution. 

\subsubsection{Synchronous Network Access Platform Model}
Another option is to assume that network access is synchronized -- i.e.,~$\forall i, \phi_i^{net}=0$ and $t_1=n\cdot p_i$ for some $n$ in \figref{endToEndTiming}. In this case, the network MILP for parameter synthesis is executed first, with only message deadlines being subject to minimization to `leave' most time for sensing and control -- resulting in the most efficient problem decomposition. On the other hand, if the ECUs MILP is run first, both sensing deadlines should be minimized and control offsets should be maximized as described in Section~\ref{subsec:syncSens}. Then, in the second stage, the ECUs' synthesis MILP is a feasibility problem, with additional simplifications since the constraints~\eqref{eq:c1}-\eqref{eq:c3} become active (i.e., equalities hold), and for all $i$, $d_i^{net}$ are pre-specified and $\phi_i^{net}=0$.
In terms of complexity, this approach is appropriate for large problems since it decouples the ECU and network analysis. Consequently, this reduces the number of variables and constraints per program since now only a part of the time testing instants remain variables.

%% file: opportunistic.tex
\section{Opportunistic Authentications} 
\label{sec:opportunistic}

The design-time framework from Section~\ref{sec:MILP} addresses Problem~1, resulting in schedulable secure control transactions  with the desired levels of QoC even in the presence of attacks. However, the overall QoC guarantees may be improved if the overall authentication rates, captured by $l_i$'s, are increased, which can be achieved if additional system resources (ECU time, network bandwidth) are available. While the QoC degradation curves capture the dependency between QoC and authentication rates (i.e., $l_i$), making the distances between authentications $l_i$ variables, instead of predefined values obtained from the QoC requirements, as part of the presented MILP  does not scale. Consequently, the methods we introduced in~\cite{lesi_tecs17,lesi_rtss17} to  optimally allocate resources in systems where only network or only ECE scheduling is considered,  such that the overall QoC under attack is maximized 
cannot be employed for systems featuring many tasks/messages when both network and task scheduling is considered. 


On the other hand, for secure transactions with periodic cumulative authentication policies $\mu_i(s_i,l_i,f_i)$ obtained by the MILP-based framework from Section~\ref{sec:MILP},
ECUs and the network will commonly not be entirely utilized at runtime. Thus, in this section we consider the problem of how intermittent authentication can be added at runtime, on top of a system for which we already obtained strong timeliness and QoC-under-attack guarantees (i.e., Problem~2). 
As our goal is to develop a runtime scheme that allocates available resources (CPU/network time) to authenticate additional sensor messages, we assume that the following holds. 
First, each ECU needs to have the knowledge of the network's busy intervals, or equivalently, of the temporal parameters of the network's workload, to ensure that additional transmitted MACs do not affect timeliness of existing periodic traffic. This is a valid assumption in low-level control networks (e.g.,~CAN bus that is considered in the case study in Section~\ref{sec:caseStudy}),
where traffic patterns are fully defined at design-time. Secondly, each ECU needs to have knowledge of its own available processing time, 
to ensure that additional MAC signing or verification can be performed without violating timing constraints of existing transactions, and other periodic and worst-case sporadic workload.
%
This is typically satisfied for constrained embedded platforms targeted by this general framework, 
as they commonly execute reservation-based RTOSs that enforce runtime timeliness guarantees. 

In such systems, our goal is to develop a runtime policy to determine the optimal, or near-optimal \emph{opportunities} for additional sensor measurements to be authenticated. In essence, this policy defines ECU-side computation of the priority level with which the specific MAC transmission will compete with other ECUs attempting to opportunistically authenticate additional sensor measurements.

Intermittent authentications should only be allowed outside the times captured by the deployed periodic cumulative authentication
 policies $\mu(s_i,l_i,f_i)$. 
 To improve the overall QoC guarantees, we consider QoC degradation curves $\mathcal{J}_i$ for every plant, and assign priority to a MAC transmission based on the level of improvement in the overall QoC that the specific authenticated measurement would contribute. Specifically, we assign a reward $r_i(t)$ at time $t$ to an opportunistic authentication~as
$$r_i(t)=\omega_i\mathcal{J}_i(\Delta l_i(t),f_i)\text{, where } \Delta l_i(t) = \floor*{\frac{\min \left({t-t_{i_{k-1}}}, t_{i_{k}}-t\right) }{p_i}},$$
where $t_{i_{k-1}}$ and $t_{i_{k}}$ are the nearest preceding and superseding  periodic authentication release times. 
This ensures that additional authentications are favored in the middle of periods of regularly scheduled authentications, as that results in tighter bounds on the attacker. Moreover, the weights~$\omega_i$ facilitate boosting priority of more important plants (e.g.,~steering over climate control). 

This approach is practical as the light-weight priority computation can be performed on the ECU itself in the case of the CAN bus, as the standard CAN protocol incorporates message priorities into the message identification field, while transmission conflicts are intrinsically resolved. Alternatively, the centralized scheduler assumed in TTCAN networks can enforce this policy, while each ECU in FlexRay networks features a \emph{bus guardian}, that enforces design-time network access patters at runtime, and can be augmented with the aforementioned functionality. In Section~\ref{sec:caseStudy}, we demonstrate how this approach can be used to significantly improve QoC under attack at runtime, at the expense of small amounts of utilized processing times and network bandwidth.


%% file: evaluation.tex

\section{Evaluation}
\label{sec:evaluation}

In this section, we evaluate our approach both on synthetic transaction sets (Section~\ref{sec:generalEvaluation}) and a realistic automotive case-study (Section~\ref{sec:caseStudy}).

\begin{figure}[!t]
	\centering
	\pgfplotsset{every axis/.append style={line width=.5pt}}
	\resizebox{.88\textwidth}{!}{\input{Figures/solverRuntime.tex}}
	\caption{Average Gurobi solver runtime and $95\%$ confidence intervals for synthetic systems with utilizations $0.1-0.9$, constructed in accordance to the guidelines from~\cite{automotiveBenchmarks2015}.}
	\label{fig:solverRuntime}
\end{figure}
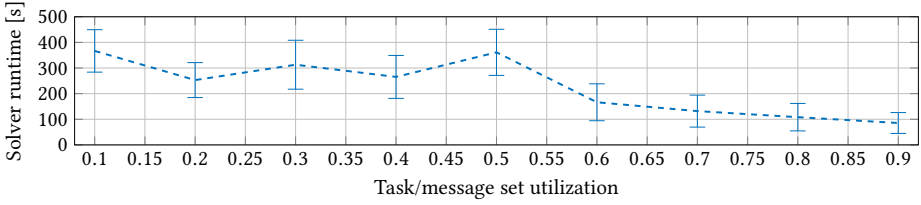
\begin{table}[!t]
\centering
\caption{Distribution of tasks and messages among periods in synthetic workloads used for generic evaluation, as well as non-QoC-related workloads used for the case study; the tasks and messages were obtained using the guidelines for automotive benchmarks from~\cite{automotiveBenchmarks2015}.}
    \begin{tabular}{c|c|c}
      \hline
      \begin{tabular}{@{}c@{}}period\\$[ms]$\end{tabular} & \begin{tabular}{@{}c@{}}share of preemptive\\(ECU) workload\end{tabular} & \begin{tabular}{@{}c@{}}share of non-preemptive\\(CAN bus) workload\end{tabular} \\ \hline
      \rowcolor[gray]{0.8} 5 & 2.5 \% & 2.63 \%\\
      10 & 31.25 \%& 32.89 \%\\
      \rowcolor[gray]{0.8} 20 & 31.25 \%& 32.89 \%\\
      50 & 3.75 \%& 3.95 \%\\
      \rowcolor[gray]{0.8} 100 & 25 \%& 26.32 \%\\
      200 & 1.25 \%& 1.32 \%\\
      \rowcolor[gray]{0.8} 1000 & 5 \%& ---\\
      \hline
    \end{tabular}
    \label{tab:workloadDistribution}
\end{table}

\subsection{Evaluation on Synthetic Systems}
\label{sec:generalEvaluation}
To evaluate general performance of our approach, we generate over $5000$ 
 synthetic systems,  each featuring 10 to 50 control transactions,  following 
the guidelines for design of automotive benchmarks from~\cite{automotiveBenchmarks2015}. Since the
guidelines focus on defining ECU-bound workloads, we redistribute the angle-synchronous workload\footnote{Angle-synchronous tasks have periods that depend on the engine speed -- i.e., the crankshaft angle determines job release.} and workloads with periods $1~ms$, $2~ms$ 
evenly to workloads with other periods. This is done for synthetic message sets 
as most practical network workloads do not include messages with such short 
periods.
Similar benchmark modifications were used in~\cite{ZengCANFD_2017}, and the resulting distribution among periods is summarized in Table~\ref{tab:workloadDistribution}.
As in~\cite{automotiveBenchmarks2015}, 
we 
scale execution times 
to assess performance under different utilization levels. Message transmission times are computed based on full-size CAN bus payload of $64~bits$ by varying transmission rate, to vary network utilization levels.  Finally, we randomly assign extended frame distances, and cumulative authentication block lengths in the range $l_i\in[1,5]$ and $f_i\in[1,3]$, while assuming that $25\%-50\%$ of tasks/messages are QoC-related (and the remaining 
workload are standard real-time tasks/messages).

We evaluate scalability of our framework by applying the decomposed MILP approach to all synthetic systems to complete the generated transaction sets. 
\figref{solverRuntime} summarizes Gurobi solver~\cite{gurobi} runtime as a function of the number of tasks/messages and task/message set utilization,\footnote{All computations are performed on a Sandy Bridge EP-based workstation with dual 3.3~GHz Intel Xeon CPUs  and 64GB of~memory.} showing applicability of our approach. 
Larger task sets typically cause longer solver runtime due to a generally larger parameter space. Relatively large variability can be attributed to random extended frame distances, which determine the hyperperiod and harmonicity of extended frame executions. Also, solver runtime is generally lower for unschedulable transaction sets regardless of the number of tasks since the solver is typically able to quickly prune large portions of the variable space which expedites conclusions about unschedulability 
-- average runtime in this case is $55~s$.

\subsection{Case Study}
\label{sec:caseStudy}

\begin{figure*}[!t]
	\begin{center}
	\subfigure [Adaptive cruise control]
	{
        \pgfplotsset{every axis/.append style={line width=0.7pt}}
        \resizebox{0.3\textwidth}{!}{\input{Figures/accCurve.tex}}
		\label{fig:AdaptiveCruise}
	}
	\subfigure  [Driveline management]
	{
        \pgfplotsset{every axis/.append style={line width=0.7pt}}
        \resizebox{0.3\textwidth}{!}{\input{Figures/dmCurve.tex}}
		\label{fig:DrivelineManag}
	}
	\subfigure [Lane keeping]
	{
        \pgfplotsset{every axis/.append style={line width=0.7pt}}
        \resizebox{0.29\textwidth}{!}{\input{Figures/lkCurve.tex}}
		\label{fig:LaneTrack}
	}
	\end{center}
\captionsetup{aboveskip=-1pt}\caption{QoC degradation curves for three considered systems --- maximal attack-induced state estimation error is bounded given a specific integrity enforcement policy determined by inter-enforcement distance~$l_i$ and authentication block length~$f_i$. Note that the adaptive cruise control system requires at least two consecutive measurements to be authenticated ($f_{ACC}^{sens}\geq2$).}
\label{fig:AllQocCurves}
\end{figure*}
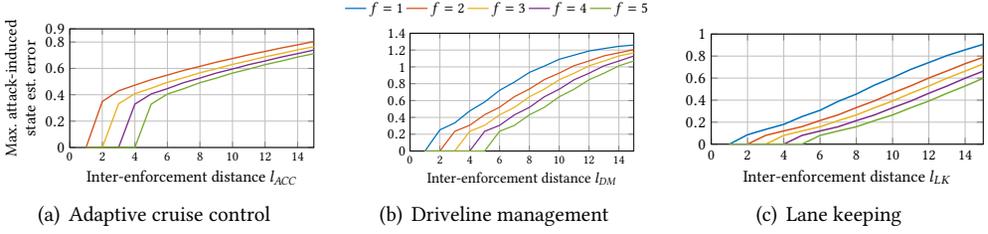
\begin{figure*}[!t]
\begin{center}
	\centering
	\pgfplotsset{every axis/.append style={line width=0.5pt}}
	\resizebox{\textwidth}{!}{\input{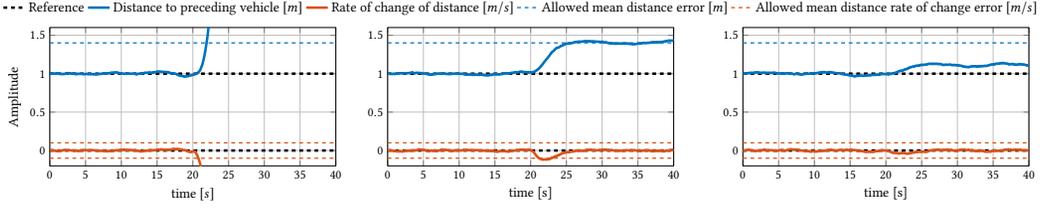}}
	\caption{Adaptive cruise control QoC under stealthy attack (starts at $t=20~s$) without integrity enforcements (left), periodic cumulative authentication with $l_{ACC}=5$ (center), and with intermittent cumulative  authentication with $\hat{l}_{ACC}=2.5$ (right).}
	\label{fig:ACCsim}
\end{center}
\end{figure*}
\begin{figure*}[!t]
\begin{center}
	\centering
	\pgfplotsset{every axis/.append style={line width=0.5pt}}
	\resizebox{\textwidth}{!}{\input{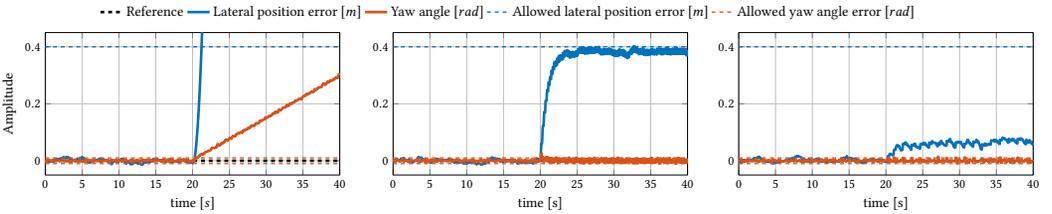}}
	\captionsetup{aboveskip=-0.2pt}\caption{Lane keeping QoC under stealthy attack (starts at $t=20~s$) without integrity enforcements (left), with a periodic cumulative  authentication with $l_{LK}=10$ (center), and with intermittent cumulative  authentication with $\hat{l}_{LK}=2.86$ (right).}
	\label{fig:LKsim}
\end{center}
\end{figure*}

We consider a realistic automotive case study where controllers for 
adaptive cruise control, lateral control for lane tracking, and driveline management, are mapped onto 3 out of 8 ECUs, with all ECUs also executing non-QoC-related workload 
as in~Table~\ref{tab:workloadDistribution}.
To model the controlled physical plants, we adopted physical system models 
from~\cite{ACCmodel},~\cite{rajamani2011vehicle},~\cite{drivelineThesis}. 
The control tasks 
are receiving sensor measurements from the eight ECUs 
communicating via a shared CAN bus. 
The network load consists of $70$ full-sized CAN frames with period distribution specified in~Table~\ref{tab:workloadDistribution}, and $8$ full-sized CAN frames carrying sensor measurements with period $p_{ACC}=p_{LK}=p_{DM}=20~ms$. As $64~bit$ MACs are used to sign sensor measurements, to ensure low probability of forgery, an entire additional frame needs to be transmitted for an authenticated measurement, as the standard CAN payload is only $64~bits$. 
%
With the standard $1 Mbps$ CAN 
rate, 
regardless of ECU utilization the system is not schedulable 
when every sensor measurement is signed.  

\figref{AllQocCurves} shows QoC degradation curves for these systems, based on which we can map admissible levels of state estimation error due to attack into computation and bandwidth requirements. Specifically, we assume that state estimation error due to attack of no more than $0.4~m$ in distance to preceding vehicle, and no more than $0.1~\frac{m}{s}$ in speed is allowed in the case of adaptive cruise control. Similarly, maximum attack-induced state estimation errors for lateral position error, its rate of change, yaw angle error and its rate of change are set to $0.4~m$, $0.1~\frac{m}{s}$, $0.01~rad$, and $0.01~\frac{rad}{s}$, respectively. Finally, drive-shaft torsion, and its rate of change state estimation errors due to attack are limited to $0.02~rad$ and $1~\frac{rad}{s}$, respectively. Thus, inter-enforcement distances and authentication block lengths resulting from these requirements are $l_{ACC}=5, f_{ACC}=3$; $l_{LK}=10, f_{LK}=2$; $l_{DM}=10, f_{DM}=1$. 

Under these conditions, in the first step of our decomposed MILP approach, Gurobi solver takes an average of $2716~s$ to return minimal deadlines for the considered message set and assign initial authentication start times such that timeliness can be guaranteed for network messages.
%
In the second step, 
for a MILP that encompasses conditions for the three control ECUs, conditioned by the previously obtained message deadlines, Gurobi takes an average of $937~s$ to complete the secure transaction set with schedulable sensing task offsets and control task deadlines. 
\figref{ACCsim} and \figref{LKsim} show the resulting trajectories for adaptive cruise control and lane keeping systems, when stealthy attacks start at $t=10~s$. 
\figref{ACCsim}~and \figref{LKsim}(left) show effects of the attack without authentication; both longitudinal and lateral control of the vehicle are entirely taken over by the stealthy attacker. \figref{ACCsim}~and~\figref{LKsim}(center) show~how the attack impact is contained within permissible limits when integrity of sensor data is enforced with the aforementioned periodic cumulative policies, resulting in network utilization of $U_{net}=0.68$.

To demonstrate benefits of using opportunistic scheduling to further improve the overall QoC under attack, we simulate additional sporadic network traffic as well as opportunistically add MACs (as described in Section~\ref{sec:opportunistic}) to sensor measurements that are not authenticated by periodic cumulative authentications. Sporadic messages are assumed to arrive with minimum inter-arrival time of $10~ms$ 
utilizing up to $5\%$ of the network bandwidth.
The resulting mean inter-authentication distance for the three systems under consideration is $\hat{l}_{ACC}=2.5$, $\hat{l}_{LK}=2.86$, and $\hat{l}_{DM}=2.31$, respectively. \figref{ACCsim}~and~\figref{LKsim}(right) show significantly improved QoC levels under attack, 
while the shared network utilization increases on average by $10\%$ due to opportunistic authentications. The final network utilization 
is $U_{net}=0.84$. ECU utilization increases on average by only $1.5\%$ to support signing and verification of additional~MACs, 
illustrating the applicability of the presented~framework.

%% file: Figures/solverRuntime.tex
%
%
\definecolor{mycolor1}{rgb}{0.00000,0.44700,0.74100}%
\definecolor{mycolor2}{rgb}{0.85000,0.32500,0.09800}%
\begin{tikzpicture}

\begin{axis}[%
width=6in,
height=0.9in,
scale only axis,
xmin=0.08,
xmax=0.92,
xticklabel style={font=\large},
xlabel style={font=\color{white!15!black}, font=\large},
xlabel={Task/message set utilization},
ymin=0,
ymax=500,
ytick={0,100,200,300,400,500},
yticklabel style={font=\large},
ylabel style={font=\color{white!15!black}, font=\large, align=center},
ylabel=Solver runtime {[s]},
axis background/.style={fill=white},
xmajorgrids,
ymajorgrids
]
\addplot [color=mycolor1, dashed, forget plot, line width=1.0pt]
  table[row sep=crcr]{%
0.1	366.767337526316\\
0.2	252.96346686\\
0.3	312.857900565217 \\
0.4	265.297659804348\\
0.5	361.182583826087\\
0.6	166.3341495\\
0.7	131.874445695652\\
0.8	108.231656630435\\
0.9	85.4749761956522\\
};
\addplot [color=mycolor1, line width=4.0pt, draw=none, forget plot, mark size=4]
 plot [error bars/.cd, y dir = both, y explicit]
 table[row sep=crcr, y error plus index=2, y error minus index=3]{%
0.1	366.767337526316	82.8247088703754	82.8247088703754\\
0.2	252.96346686	68.1756321426797	68.1756321426797\\
0.3	312.857900565217	95.4246962498722	95.4246962498722\\
0.4	265.297659804348	83.7461835152945	83.7461835152945\\
0.5	361.182583826087	89.9586	89.9586\\
0.6	166.3341495	72.0221374110238	72.0221374110238\\
0.7	131.874445695652	62.4186894832417	62.4186894832417\\
0.8	108.231656630435	53.5610784812208	53.5610784812208\\
0.9	85.4749761956522	40.7303872536321	40.7303872536321\\
};
\end{axis}
\end{tikzpicture}%

%% file: Figures/accCurve.tex
%
%
\definecolor{mycolor1}{rgb}{0.00000,0.44700,0.74100}%
\definecolor{mycolor2}{rgb}{0.85000,0.32500,0.09800}%
\definecolor{mycolor3}{rgb}{0.92900,0.69400,0.12500}%
\definecolor{mycolor4}{rgb}{0.49400,0.18400,0.55600}%
\definecolor{mycolor5}{rgb}{0.46600,0.67400,0.18800}%
\begin{tikzpicture}

\begin{axis}[%
width=2.5in,
height=1.2in,
scale only axis,
xmin=0,
xmax=15,
xlabel style={font=\color{white!15!black}, font=\Large},
xlabel={Inter-enforcement distance $l_{ACC}$},
ymin=0,
ymax=0.9,
ylabel style={align=center, font=\color{white!15!black}, font=\Large},
ylabel=Max. attack-induced\\state est. error,
yticklabel style={/pgf/number format/fixed, /pgf/number format/precision=3, font=\Large},
ytick={0,0.2,0.4,0.6,0.8,0.9},
scaled y ticks=false,
axis background/.style={fill=white},
xmajorgrids,
ymajorgrids,
]
\addplot [color=mycolor2, line width=1.0pt]
  table[row sep=crcr]{%
1	0\\
2	0.34904\\
3	0.42962\\
4	0.47273\\
5	0.51413\\
6	0.54932\\
7	0.58485\\
8	0.61551\\
9	0.6471\\
10	0.67468\\
11	0.70338\\
12	0.72864\\
13	0.75511\\
14	0.77854\\
15	0.80323\\
};

\addplot [color=mycolor3, line width=1.0pt]
  table[row sep=crcr]{%
1	0\\
2	0\\
3	0.33109539\\
4	0.40847539\\
5	0.44932539\\
6	0.49432539\\
7	0.52897539\\
8	0.56670539\\
9	0.59731539\\
10	0.63044539\\
11	0.65812539\\
12	0.68797539\\
13	0.71339539\\
14	0.74077539\\
15	0.76439539\\
};

\addplot [color=mycolor4, line width=1.0pt]
  table[row sep=crcr]{%
1	0\\
2	0\\
3	0\\
4	0.32880539\\
5	0.40578539\\
6	0.44632539\\
7	0.49187539\\
8	0.52646539\\
9	0.56453539\\
10	0.59514539\\
11	0.62848539\\
12	0.65618539\\
13	0.68620539\\
14	0.71165539\\
15	0.73917539\\
};

\addplot [color=mycolor5, line width=1.0pt]
  table[row sep=crcr]{%
1	0\\
2	0\\
3	0\\
4	0\\
5	0.32849539\\
6	0.40542539\\
7	0.44592539\\
8	0.49156539\\
9	0.52614539\\
10	0.56427539\\
11	0.59487539\\
12	0.62826539\\
13	0.65596539\\
14	0.68602539\\
15	0.71147539\\
};

\end{axis}
\end{tikzpicture}%

%% file: Figures/dmCurve.tex
%
%
\definecolor{mycolor1}{rgb}{0.00000,0.44700,0.74100}%
\definecolor{mycolor2}{rgb}{0.85000,0.32500,0.09800}%
\definecolor{mycolor3}{rgb}{0.92900,0.69400,0.12500}%
\definecolor{mycolor4}{rgb}{0.49400,0.18400,0.55600}%
\definecolor{mycolor5}{rgb}{0.46600,0.67400,0.18800}%
\begin{tikzpicture}

\begin{axis}[%
width=2.5in,
height=1.3in,
scale only axis,
xmin=0,
xmax=15,
xlabel style={font=\color{white!15!black}, font=\Large},
xlabel={Inter-enforcement distance $l_{DM}$},
ymin=0,
ymax=1.4,
ylabel style={align=center, font=\color{white!15!black}, font=\Large},
yticklabel style={/pgf/number format/fixed, /pgf/number format/precision=3, font=\Large},
ytick={0,0.2,0.4,0.6,0.8,1.0,1.2,1.4},
axis background/.style={fill=white},
xmajorgrids,
ymajorgrids,
legend style={at={(1.1,1.1)}, anchor=south east, legend cell align=left, align=left, draw=white!15!black, draw = none, legend columns = 5, line width=3.0pt, font=\Large}
]
\addplot [color=mycolor1, line width=1.0pt]
  table[row sep=crcr]{%
1	0\\
2	0.25131\\
3	0.33449\\
4	0.47678\\
5	0.58235\\
6	0.72102\\
7	0.82008\\
8	0.93428\\
9	1.0096\\
10	1.0904\\
11	1.1396\\
12	1.1891\\
13	1.2177\\
14	1.2445\\
15	1.2594\\
};
\addlegendentry{$f=1$}

\addplot [color=mycolor2, line width=1.0pt]
  table[row sep=crcr]{%
1	0\\
2	0\\
3	0.234563\\
4	0.306023\\
5	0.431463\\
6	0.520453\\
7	0.646373\\
8	0.736793\\
9	0.848953\\
10	0.926223\\
11	1.014053\\
12	1.071153\\
13	1.130853\\
14	1.167753\\
15	1.203853\\
};
\addlegendentry{$f=2$}

\addplot [color=mycolor3, line width=1.0pt]
  table[row sep=crcr]{%
1	0\\
2	0\\
3	0\\
4	0.234063\\
5	0.305043\\
6	0.430143\\
7	0.518773\\
8	0.644513\\
9	0.734803\\
10	0.847003\\
11	0.924353\\
12	1.012353\\
13	1.069653\\
14	1.129653\\
15	1.166753\\
};
\addlegendentry{$f=3$}

\addplot [color=mycolor4, line width=1.0pt]
  table[row sep=crcr]{%
1	0\\
2	0\\
3	0\\
4	0\\
5	0.234053\\
6	0.305043\\
7	0.430143\\
8	0.518773\\
9	0.644513\\
10	0.734793\\
11	0.846993\\
12	0.924343\\
13	1.012353\\
14	1.069653\\
15	1.129653\\
};
\addlegendentry{$f=4$}

\addplot [color=mycolor5, line width=1.0pt]
  table[row sep=crcr]{%
1	0\\
2	0\\
3	0\\
4	0\\
5	0\\
6	0.234053\\
7	0.305033\\
8	0.430143\\
9	0.518773\\
10	0.644513\\
11	0.734793\\
12	0.846993\\
13	0.924343\\
14	1.012353\\
15	1.069653\\
};
\addlegendentry{$f=5$}

\end{axis}
\end{tikzpicture}%

%% file: Figures/lkCurve.tex
%
%
\definecolor{mycolor1}{rgb}{0.00000,0.44700,0.74100}%
\definecolor{mycolor2}{rgb}{0.85000,0.32500,0.09800}%
\definecolor{mycolor3}{rgb}{0.92900,0.69400,0.12500}%
\definecolor{mycolor4}{rgb}{0.49400,0.18400,0.55600}%
\definecolor{mycolor5}{rgb}{0.46600,0.67400,0.18800}%
\begin{tikzpicture}

\begin{axis}[%
width=2.5in,
height=1in,
scale only axis,
xmin=0,
xmax=15,
xlabel style={font=\color{white!15!black}, font=\large},
xlabel={Inter-enforcement distance $l_{LK}$},
ymin=0,
ymax=1,
ylabel style={align=center, font=\color{white!15!black}, font=\large},
yticklabel style={/pgf/number format/fixed, /pgf/number format/precision=3, font=\large},
axis background/.style={fill=white},
xmajorgrids,
ymajorgrids,
]
\addplot [color=mycolor1, line width=1.0pt]
  table[row sep=crcr]{%
1	0\\
2	0.085825\\
3	0.13434\\
4	0.18008\\
5	0.25041\\
6	0.30783\\
7	0.38884\\
8	0.4549\\
9	0.53617\\
10	0.60413\\
11	0.67874\\
12	0.74128\\
13	0.80452\\
14	0.85663\\
15	0.90659\\
};

\addplot [color=mycolor2, line width=1.0pt]
  table[row sep=crcr]{%
1	0\\
2	0\\
3	0.07990063\\
4	0.12062763\\
5	0.15914763\\
6	0.21512763\\
7	0.26657763\\
8	0.33229763\\
9	0.39308763\\
10	0.46419763\\
11	0.52941763\\
12	0.60073763\\
13	0.66483763\\
14	0.73081763\\
15	0.78877763\\
};

\addplot [color=mycolor3, line width=1.0pt]
  table[row sep=crcr]{%
1	0\\
2	0\\
3	0\\
4	0.07960863\\
5	0.11995763\\
6	0.15837763\\
7	0.21394763\\
8	0.26540763\\
9	0.33075763\\
10	0.39164763\\
11	0.46244763\\
12	0.52781763\\
13	0.59890763\\
14	0.66320763\\
15	0.72902763\\
};

\addplot [color=mycolor4, line width=1.0pt]
  table[row sep=crcr]{%
1	0\\
2	0\\
3	0\\
4	0\\
5	0.07960663\\
6	0.11994763\\
7	0.15837763\\
8	0.21393763\\
9	0.26540763\\
10	0.33074763\\
11	0.39163763\\
12	0.46242763\\
13	0.52779763\\
14	0.59887763\\
15	0.66316763\\
};

\addplot [color=mycolor5, line width=1.0pt]
  table[row sep=crcr]{%
1	0\\
2	0\\
3	0\\
4	0\\
5	0\\
6	0.07960563\\
7	0.11994763\\
8	0.15837763\\
9	0.21393763\\
10	0.26540763\\
11	0.33074763\\
12	0.39163763\\
13	0.46242763\\
14	0.52779763\\
15	0.59886763\\
};

\end{axis}
\end{tikzpicture}%

%% file: relatedWork.tex
\section{Related Work}
\label{sec:relatedWork}

Integrating security guarantees into legacy and resource-constrained systems has attracted significant research attention.
For instance, in~\cite{OpportunisticExecutionLegacyRT} the authors explore opportunistic execution of security services in legacy real-time systems, while leveraging hierarchical scheduling to ensure that schedulability of existing tasks is not impaired. The security performance metric proposed therein is the frequency of periodic execution of security services. In~\cite{ImprovingEmbeddedSecurity}, a novel scheduling policy is proposed for embedded systems to ensure schedulability of real-time control tasks subject to both timing and security constraints. This is achieved by optimal distribution of slack times which are computed after schedulability of existing control tasks is guaranteed. Among a variety of security services, an optimal schedule is constructed based on abstract relative security levels. 
In~\cite{StaticRTsecurity}, the authors devise a security-aware EDF schedulability test. Therein, security services are grouped by security level and execution of security services from different groups is combined to increase Quality-of-Security (e.g., message encryption can be combined with authentication to protect both confidentiality and integrity of transmitted data). Consequently, group-based security model is integrated with EDF scheduling and a security-aware optimization problem is formulated around scheduling of suitable security services given a set of real-time tasks. 
However, it is important to highlight that no existing work provides a direct relationship between resource utilization and actual systems' performance pertaining to its main functionality (i.e., control performance, Quality-of-Control) -- in fact, only abstract \emph{security levels} are considered.

Transaction scheduling is typically considered separately for time- and event-triggered communication models. Event-triggered transaction scheduling requires additional overheads for event signalling, i.e.,~synchronization between the transmitting and receiving nodes is explicitly obtained by transmission of additional messages. Examples of works addressing analysis of such transaction implementation schemes are~\cite{distrSync1,distrSync2}. For systems where network traffic patterns are determined by design, and resources (both processor computation power and network bandwidth) are severely constrained, satisfaction of timing constraints for transactions can be achieved by careful offset/dedline enforcement --- the approach considered in this paper. Traditional offset-based schedulability analysis for distributed systems under rate monotonic were presented in the original \emph{holistic analysis} framework from~\cite{holisticDistributedRM}, and further improved in~\cite{staticDynamicOffsets,bcWCdistributedRT}. Furthermore, this analysis has been extended to EDF in~\cite{spuriEDFdistr1}. However, only the standard task models are observed, and these works mostly focus on computing response times while no optimization framework is devised to generate feasible offsets (or deadlines). In~\cite{feasibleDeadlines}, the authors develop a technique to compute a (sufficient) region of admissible deadlines given a set of tasks under EDF, which enables the designer to optimize the desired performance metric. However, this approach is non-trivial to integrate into an end-to-end schedulability analysis framework, due to its recursive algorithmic nature.

%% file: conclusion.tex

\section{Conclusion}
\label{sec:conclusion}

In this paper, we have presented a MILP-based framework for integrating security guarantees with end-to-end timeliness requirements for control transactions in resource-constrained CPS. We have shown that the use of physics-based anomaly/intrusion detectors and intermittent message authentication results in strong QoC performance guarantees in the presence of network-based attacks without significant security-related resource overhead. 
We have also shown how the security-related overhead can be additionally reduced with the use of cumulative authentication policies, which can be implemented such that real-time guarantees for control-related tasks and messages are retained, while QoC in the presence of attacks is maintained within the permissible design-time limits. In addition, we have presented a method to integrate intermittent authentication policies in a near-optimal manner from the QoC standpoint,  to opportunistically exploit available processor time and network bandwidth at runtime. As our approach fully supports cumulative authentication policies, it can be used for dynamical systems where solely authenticating a single sensor measurement periodically or intermittently is not sufficient to provide QoC guarantees under attack. Finally, for large-scale systems where a unified scheduling approach for all ECUs and network may  be intractable, we have shown how the problem can be decomposed in a platform/implementation-specific manner. We have demonstrated scalability and effectiveness of our approach on both synthetic systems and a realistic automotive case study and shown that security guarantees can be incorporated without violating existing timeliness properties even with limited resource availability.